\documentclass[11pt,english]{article}
\pdfoutput=1
\usepackage[latin9]{inputenc}
\usepackage{amsmath,amssymb,amsthm}
\usepackage{graphicx}
\usepackage{setspace}
\usepackage{babel}
\usepackage{xcolor}
\usepackage{comment}

\usepackage[letterpaper]{geometry}
\theoremstyle{plain}
\newtheorem{theorem}{Theorem}
\newtheorem*{theorem*}{Theorem}

\newtheorem*{corollary*}{Corollary}
\theoremstyle{definition}
\newtheorem{remark}{Remark}
\newtheorem*{remark*}{Remark}

\def\O{\Omega}
\def\s{\sigma}
\def\l{\lambda}

\def\pfi{\varphi}
\def\dphi{\dot{\phi}}
\def\ddphi{\ddot{\phi}}
\def\dth{\dot{\theta}}
\def\ddth{\ddot{\theta}}
\def\dpsi{\dot{\psi}}
\def\ddpsi{\ddot{\psi}}
\def\dchi{\dot{\chi}}
\def\ddchi{\ddot{\chi}}
\def\dxi{\dot{\xi}}
\def\ddxi{\ddot{\xi}}
\def\o{\omega}

\def\d{\delta}

\def\ds{\displaystyle}

\makeatletter

\date{}

\makeatother

\begin{document}

\title{Swinging a playground swing\\ \Large Torque controls for inducing sustained oscillations}

\author{Sergiy Koshkin\\
 Department of Mathematics and Statistics\\
 University of Houston-Downtown\\
 One Main Street, S-713\\
 Houston, TX 77002\\
 e-mail: koshkins@uhd.edu \and
 Vojin Jovanovic\\
 Customer Performance Center\\
 Smith Bits, A Schlumberger Co.\\
 1310 Rankin Road\\
 Houston, TX 77073\\
 e-mail: fractal97@gmail.com}
\maketitle
\begin{abstract}
Models of a playground swing have been studied since the 1960s. However, in most of them, the position of the swinger is controlled directly. This simplifies the problem but hides the mechanics of torques applied to keep the swing moving in a regular pattern. This article studies these mechanics. Two models of a swing with torques as controls that we consider are identical to popular models of modern robotics: the Acrobot and reaction wheel pendulum. However, the control task of sustaining the swing's regular oscillations by a static feedback control is new and challenging, especially when damping in the joint connecting the swing to the frame is considered. We develop two types of controls to accomplish this task. One works for small damping and is based on linearizing the undamped system by a suitable preliminary feedback control. The other works for large damping. In the steady state, the resulting closed-loop system describes a harmonically driven damped pendulum (a simple system known for its complex behavior), including chaotic motion for some parameter values. To address such complexities, we build free parameters into the controls, then adjust them based on simulations to avoid chaos and achieve regular oscillations that are seen on playgrounds.
\medskip

\textbf{Keywords}: underactuated system, orbital stabilization, sustained oscillations, limit cycle, static state feedback, feedback linearization, bounded zero dynamics, swing up, Acrobot, reaction wheel pendulum, driven damped pendulum
\bigskip

\textbf{MSC}: 70K05 34C15 70K40 70E17 70Q05 93D25 93C85
\end{abstract}

\newpage

\section{Introduction}

The main goal of this paper is to design torque controls for a common model of a playground swing, the model of swinging in the sitting position proposed by Case and Swanson \cite{CS}, with damping added for more realism. As it turns out, the addition of damping makes the problem much more complex from the analytic standpoint. Modeling human control of playground swings has venerable history going back at least to 1960s, see \cite{CS, WRR} and references therein. However, even aside from the fact that damping or friction are rarely included, all models in the literature, to the best of our knowledge, take the  body's position as the control input. For swinging in the sitting position it is described by the angle between the body and the swing's axis (rods, chains or cables attaching the seat to the frame). This is a natural simplification given that human swingers can use their biophysical machinery to generate torques for muscle movements that change and stabilize the requisite body angle. But much complexity is hidden in these built-in mechanisms, and if we want to understand them, or design an electromechanical device that would relieve parents of the chore of swinging their kids, then we can not rely on biophysics to do the work. 

Aside from the recreational motivation, torque control of the swing is closely related to control tasks in modern robotics, such as bipedal walking and brachiating for robots with rotating joints \cite{Asan,NFK,RbtC,Tzaf}. Simplified models of robots, such as the Acrobot and the Pendubot, have been developed to design and analyze controls of similar nature \cite[ch. 5]{FL},\cite{Spong,KCS,Mur}, as well as component models such as the reaction wheel pendulum \cite{SCL,ACK,Grit,ZCM} and Furuta pendulum \cite{FSGA}. When there is no damping in the swing's joints, the equations of motion of the Case-Swanson model with the torque control are identical (up to the signs of constants) to those of the Acrobot, and for the simplified model that we focus on (the balanced swing) to those of the reaction wheel pendulum. Another related task is synchronizing a pair of two connected pendulums \cite{Bl1,JK}, where one is equipped with clock escapement. 

These systems share common features of underactuation and trigonometric nonlinearity that make them convenient test cases for the methods of nonlinear geometric control theory \cite{FL,Isid,Resp}. Underactuation means that only some degrees of freedom are directly affected by control inputs, which forces one to balance achieving desired behavior with maintaining reasonable internal dynamics of the uncontrolled degrees of freedom. In robotics, underactuation comes up when some joints do not have an associated motor to move them. In human execution, performing such tasks is often associated with dexterity. In part, we use the problem of playground swinging to investigate what sort of controls an underactuated robot would require for some tasks easily performed by humans.

In the Case-Swanson model the body is represented by a rigid dumbbell made of three masses attached to a rigid axis that connects it to a fixed pivot. The axis and the dumbbell can rotate around the upper and the lower joint, respectively. We add damping to both joints, but it is adding it to the unactuated upper joint that complicates the problem considerably. The open loop equations of motion are formally identical to those of the damped double pendulum, and, with the torque control, to those of the damped Acrobot. For analytic simplicity, in most of the paper we consider the case of the \noindent\textbf{\textit{balanced swing}}, where the moments of inertia of the side masses are equal, a variant singled out already in \cite{CS}. The equations of motion reduce to those of the damped reaction wheel pendulum. Unbalanced swinging in the sitting position can be treated as a small perturbation of the balanced case (Section ``Unbalanced swing"). 

While the control system is familiar from robotics, the control task is unusual. It is not stabilization near an equilibrium, but rather inducing sustained oscillations of  prescribed amplitude and frequency (by a static state feedback control law). Sometimes such a task is interpreted as a version of trajectory tracking for a selected periodic trajectory  \cite{CS,ACK,FSGA,Grit}, but this artificially fixes the phase of oscillations that a swinger does not care about and would not want to spend energy maintaining. Therefore, a more natural control goal is {\it orbital stabilization}, where the orbit is a stable limit cycle \cite{AIb,Hak12,Hak17}. Practical considerations lead to further {\it design criteria} that the desired control must satisfy:
\begin{enumerate}
\item In the closed loop system the angle of the swing axis with the vertical approaches a stable limit cycle of oscillations with a single dominant frequency, without beats or chaotic motion.

\item The dynamics of the body angle with the swing axis is bounded and oscillatory with a single dominant frequency.

\item The transient motion, before the limit cycle is reached, is reasonably short.
\end{enumerate}
Meeting all three desiderata turns out to be surprisingly challenging. The first two criteria are hard to express in analytic form, and we are not aware of general results for nonlinear systems that would guarantee meeting them. Accordingly, our design strategy is semi-analytic. We look for a control ansatz that meets part of our design criteria and contains sufficiently many parameters that are adjusted in simulations to meet the rest. Such hybrid approach is typical for studying systems that are capable of chaotic behavior, like harmonically driven damped pendulum \cite{BB,BG}. Our adaptation of it to control design that sustains regular oscillations offers a novel methodology for controlling underactuated mechanical systems with friction that appear in robotics.

Most of the cited work on control of robot arms and pendula concerns the swing up task and stabilization near the unstable equilibrium. Many classical control techniques employed there, such as linearization around a stable equilibrium, pole placement and standard stabilizing controllers such as LQR (including its variants for trajectory tracking), are not suitable for our control task of orbital stabilization. However, the technique of feedback linearization \cite{Isid,Resp}, that has been applied to control of swing up in \cite{Spong, Mur} and of the {\it undamped} reaction wheel pendulum in \cite{SCL,Griz}, is central to our approach. Our control design is closest to that of \cite{AIb}, that uses dynamic state feedback construction to induce limit cycles with bounded zero dynamics in feedback linearizable multi-dimensional systems. Its simplification to a static state feedback form, and application to control of swinging appear to be new. Generally speaking, to make use of feedback linearization one has to get lucky twice. First, the studied system has to be feedback linearizable, or at least close to one in some sense. And second, the new variables have to be related to the original ones in a way that preserves desirable behavior. This second condition is not met in our case by an alternative construction of \cite{Hak17} based on backstepping. It is one of the main observations of this paper that we do get lucky with the balanced swing.
\medskip

The paper is organized as follows. In Section ``Playground swing: equations of motion and naive controls" we show how several natural ideas lead to controls that fail at least one of the design criteria. It also serves as preliminaries, introducing the model of the swing along with some useful background and terminology.

In Section ``Feedback linearization and imposing a limit cycle" we generalize the standard construction of a static state feedback law that imposes a limit cycle on a single second order equation \cite{Hak12,Iwa} to a linear system of any dimension in the Brunovsky canonical form. It is motivated by the dynamic state feedback construction of \cite{AIb}. The law imposes a limit cycle on a two-dimensional subsystem, while producing bounded zero dynamics on a complementary subsystem. Since the system is linear this dynamics is oscillatory. 

In Section ``Feedback linearization control: analysis and variations" we analyze the resulting control law for the balanced swing. In principle, our construction includes an arbitrary energy function that determines the shape of the limit cycle, but we always chose it to be quadratic for simplicity, and did not investigate closed loop behavior when it is varied beyond shifting the center of the cycle. So our limit cycle is an ellipse and we have five control parameters to adjust: two describe the ellipse, and three control gains determine transient behavior. Using simulations, we adjust the parameter values to meet our design criteria for realistic playground swings and swingers. Shifting the center of the ellipse is used to incorporate an additional constraint on the range of the body's angle, which is especially salient for swings with back benches that prevent backward leaning. In simulations, the zero damping law (with minor corrections) turns out to work quite well even when the damping is non-zero but relatively small. But the damped system is no longer feedback linearizable, and the resulting closed loop system is analytically opaque. 

To remedy the situation somewhat, in Section ``Damping corrections for small amplitudes" we linearize the balanced swing for small amplitudes and run the process of reduction to the Brunovsky canonical form and imposing the limit cycle on it in parallel to the non-linear zero damping case. This allows us to obtain damping corrections analytically, and to track what happens to the closed loop system (for small amplitudes) when the damping increases. The linearized closed loop system can be transformed into an alternative form that makes the physics behind the control more transparent. One can think of it as imposing a limit cycle on a combination of the swing and the body angles and their velocities designed to periodically drive the swing angle, while adding artificial damping (and stiffness) to a complementary degree of freedom. For the linearized system this combination is a linear combination that can be explicitly written for all values of damping and control parameters. But for the non-linear system it exists (in exact form) only in the case of zero damping (as we already know), and also in one other case. 

This second case is of independent interest because it does not involve feedback linearization and a small parameter perturbation. We study it in Section ``From balanced swing to driven damped pendulum". Physically, the control can be interpreted as imposing a limit cycle on a linear combination of the angles only, without involving velocities, but the analysis only works for a special choice of the ratio of the last two control gains. The resulting closed loop system for the swing angle turns out to be (in the steady state) the damped pendulum driven by a harmonic excitation. As is well known, this is one of the simplest systems with chaotic behavior, but its appearance is not entirely unexpected since the double pendulum itself is already such a system. If the driving frequency or the driving amplitude is increased, while all other parameters are kept fixed, the pendulum dynamics undergoes a cascade of period doubling bifurcations that ends in regions of chaotic behavior interspersed with returns to regular oscillations at yet higher values \cite[6.3]{BB}, \cite[3.3.5]{BG}. 

Obviously, meeting our second design criterion requires avoiding chaos, or even period doubling, and having one less control gain makes the task more difficult. It turns out that for realistic values of lengths and masses this is only possible when the damping in the upper joint is relatively large to begin with. We estimate the requisite range of damping values using the results of \cite{Mil}. Outside of this range in simulations we encountered long transient times, oscillations without dominant frequency (beats), and sometimes numerical instability that indicates sensitivity to initial values characteristic of chaotic dynamics. 

In Section ``Unbalanced swing" we represent the unbalanced swing as a perturbation of the balanced swing, and show, in simulations, that balanced controls continue to perform well for small imbalances. Finally, in Conclusions we summarize our results, and discuss their limitations and directions of future work they point to. The upshot is that  the ``well-known" playground swing still holds a lot of unanswered questions.

\section{Playground swing: equations of motion and naive controls}\label{Naive}

There is a general agreement that two mechanisms contribute to keeping a playground swing going: moving the body's center of mass up and down and rotating it around the pivot (angular momentum transfer). In idealized models of swinging these two mechanisms are often studied separately, as swinging in the standing and in the sitting position, respectively. We will start by considering the model of a swing with a human in the sitting position proposed by Case and Swanson \cite{CS}.

\subsection{Case-Swanson model}

Human body is modeled by a rigid dumbbell with three masses, and the middle mass is attached to the end of a damped pendulum, see Figure \ref{CSswing}. The angle of the swing axis with the vertical is denoted $\phi$, and the angle of the dumbbell with the axis is $\theta$.
\begin{figure}[!ht]
\centering
\includegraphics[width=2.0in]{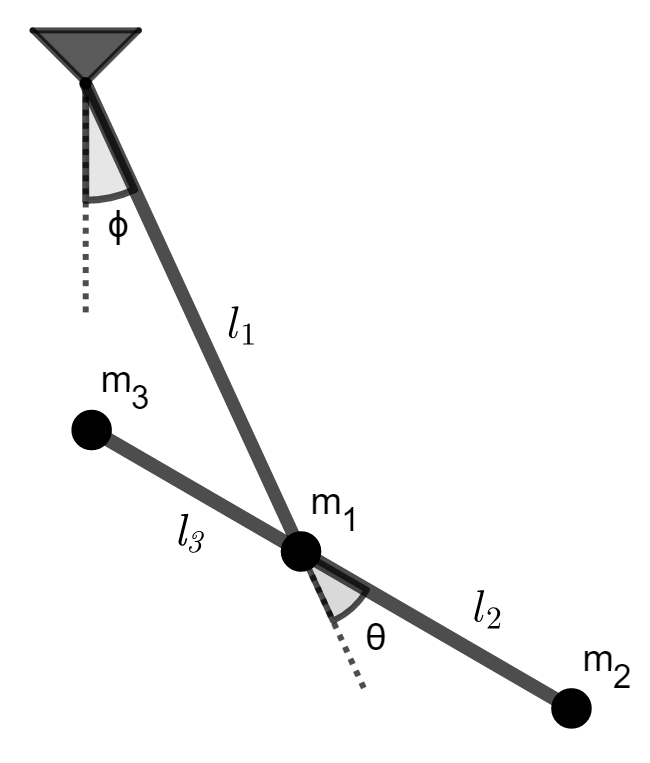}
\vspace{-1em}
\caption{\label{CSswing}A three mass model of a swing with a human in a seated position.} 
\end{figure}
Following \cite{CS}, we introduce the combined mass $M:=m_1+m_2+m_3$, the moments of inertia $I_1:=Ml_1^2$ and $I_2:=m_2l_2^2+m_3l_3^2$, and the imbalance coefficient $N=m_2l_2-m_3l_3$. Then the Lagrangian of the system is
\begin{equation}\label{CSLagr}
L=\frac12 I_1\dot{\phi}^2+\frac12 I_2(\dot{\phi}+\dot{\theta})^2+Nl_1\cos\theta\,\dot{\phi}(\dot{\phi}+\dot{\theta})+Mgl_1\cos\phi+Ng\cos(\phi+\theta).
\end{equation}
Note that the form of the Lagrangian is identical to the form of it for the \noindent\textbf{\textit{double pendulum}}, which is obtained by setting $m_3=0$. The only difference is that for a double pendulum the imbalance $N$ can only be non-negative, whereas for the swing it can have any sign depending on the size and location of dumbbell masses, i.e. the positioning of the body. 

Our control variable is not the angle $\theta$, as in most simple models, but rather the torque $u$ in the lower joint. The Lagrangian does not take into account damping, which, as we shall see, complicates the problem considerably. We introduce the damping coefficients $\alpha$ and $\beta$ in the upper and the lower joint, respectively, and modify the Euler-Lagrange equations as follows:
\begin{align}\label{ELDamp}
\begin{cases}
\frac{d}{dt}\frac{\partial L}{\partial\dot{\phi}}-\frac{\partial L}{\partial\phi}+\alpha\dot{\phi}=0\\
\frac{d}{dt}\frac{\partial L}{\partial\dot{\theta}}-\frac{\partial L}{\partial\theta}+\beta\dot{\theta}=u.
\end{cases}
\end{align} 
Throughout most of the paper we assume that the body is perfectly balanced, i.e. $N=0$. This isolates the angular momentum transfer as the \noindent\textbf{\textit{energy pumping mechanism}} and simplifies the equations of motion. The unbalanced swing is briefly considered in Section ``Unbalanced swing". For $N=0$ \eqref{ELDamp} reduces to
\begin{align}\label{BalSw}
\begin{cases}
(I_1+I_2)\ddphi+I_2\ddth+\alpha\dphi+Mgl_1\sin\phi=0\\
I_2\ddphi+I_2\ddth=u-\beta\dth.
\end{cases}
\end{align}  
When $\alpha=\beta=0$ one easily recognizes in \eqref{BalSw} the equations of motion of the \noindent\textbf{\textit{reaction wheel pendulum}} introduced in \cite{SCL}, see also \cite{KCS,Grit,ZCM,AIb,Hak12,Griz}.
It will be convenient for us to transform \eqref{BalSw} into a simpler form. First, we eliminate $\ddth$ from the first equation in \eqref{BalSw} using the second:
$$
I_1\ddphi+\alpha\dphi+Mgl_1\sin\phi=-(u-\beta\dth).
$$
Next, by introducing new constants and variables:
\begin{equation}\label{Relab}
c:=\frac{\alpha}{I_1};\ \ \ \o^2:=\frac{Mgl_1}{I_1}=\frac{g}{l_1};\ \ \ \psi:=-\frac{I_2}{I_1}(\phi+\theta);\ \ \ v:=-\frac{u-\beta\dth}{I_1},
\end{equation}  
we reduce the equations of motion to:
\begin{align}\label{psiBalSw}
\begin{cases}
\ddphi+c\dphi+\o^2\sin\phi=v\\
\ddpsi=v.
\end{cases}
\end{align}

\subsection{Damping canceling control} 

The first equation describes a driven damped pendulum, and the first thought that comes to mind for inducing sustained oscillations of it is to simply cancel the damping in the upper joint by setting $v=c\dphi$, or equivalently $u=-\alpha\dphi+\beta\dth$. With it, $\phi$ will follow the dynamics of a simple undamped pendulum that oscillates at the natural frequency $\o=\sqrt{\frac{g}{l_1}}$. Unfortunately, if one is interested in swinging with a particular amplitude this will not work. Natural oscillations are neutral rather than stable, and any perturbation will shift the amplitude up or down. In other words, the cancelling control fails already our first design criterion: swing oscillations are not attracted to a stable limit cycle.
\begin{figure}[!ht]
\centering
(a) \includegraphics[width=3in]{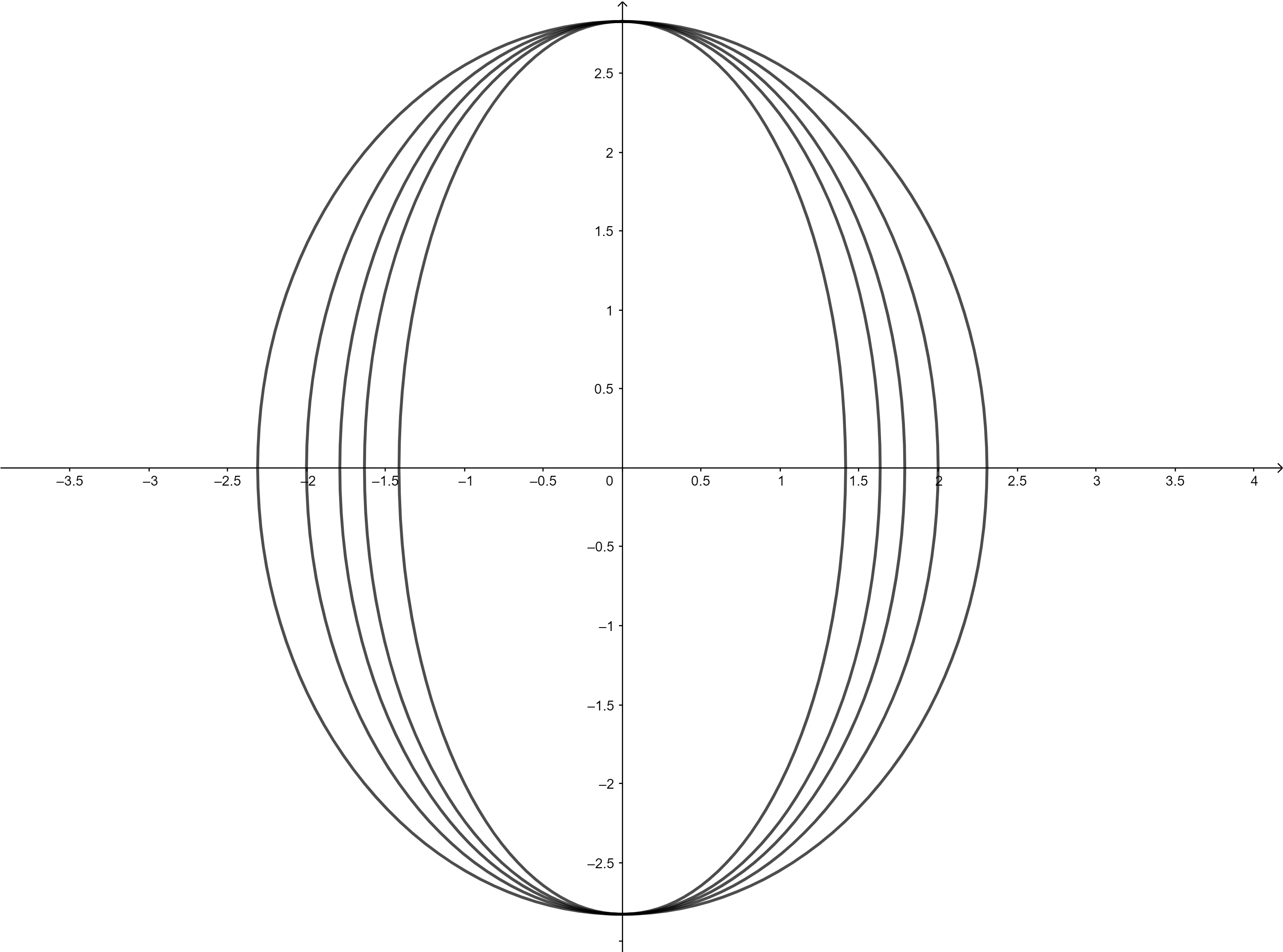}
\hspace{2em}
(b) \includegraphics[width=2.35in]{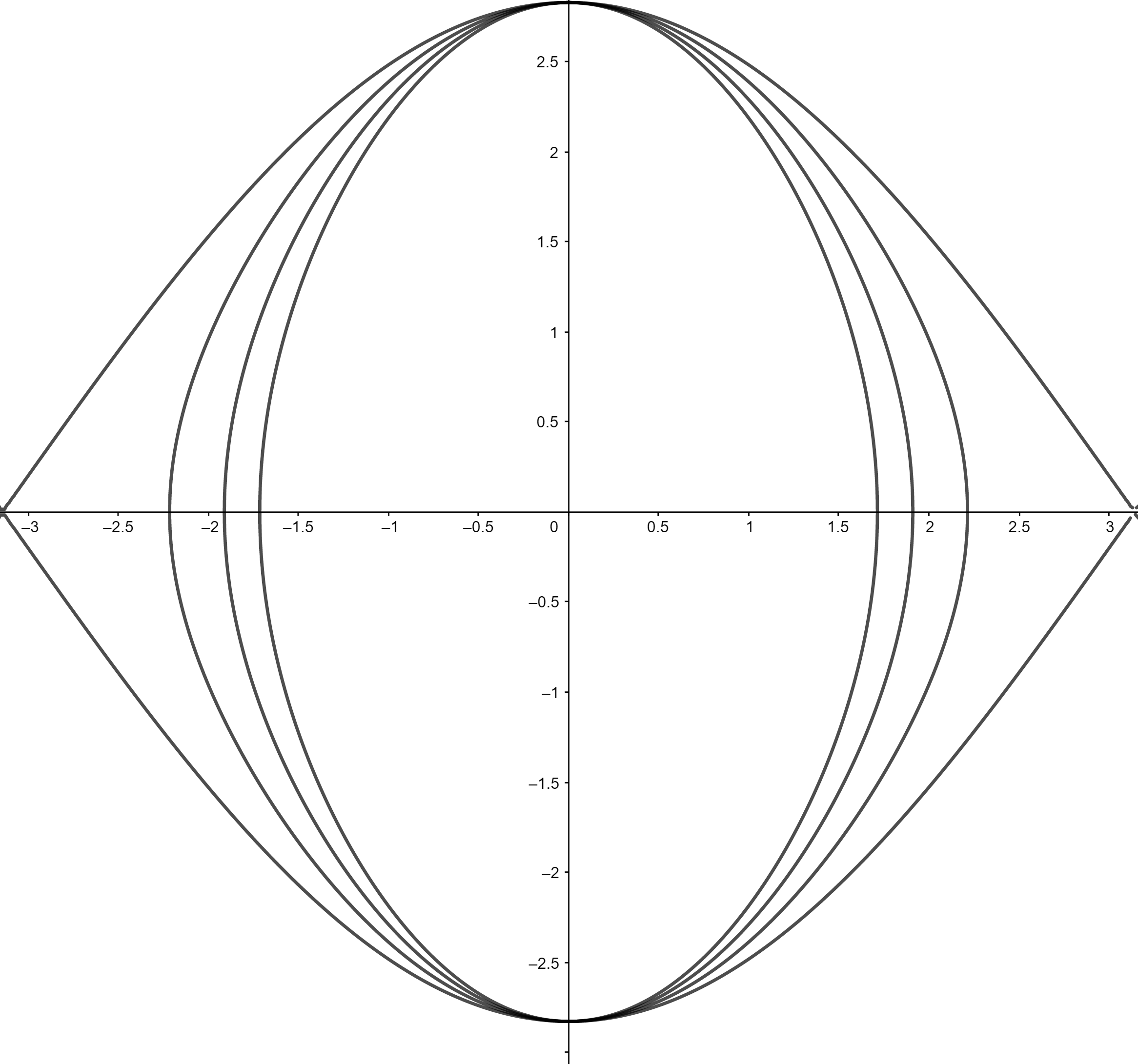}
\caption{\label{PhaseLimCyc} Phase portraits of limit cycles in the $\phi$-$\dphi$ plane for (a) $\frac12\dphi^2+\frac12\o^2\phi^2=r^2$ with $\o^2=1.5, 2, 2.5, 3, 4$; (b) $\frac12\dphi^2+\frac12\o^2(1-\cos\phi)=r^2$ with $\o^2=8,10,12,14$. In both plots $r^2=8$.} 
\end{figure} 
\medskip

\subsection{Trajectory tracking} 

This popular approach is pursued in \cite{CS,ACK,FSGA,Grit}, e.g. by prescribing $\phi(t)=A\cos\o t$ as the target trajectory, is also unsatisfactory. While the canceling control does not do enough, a tracking control does too much, prescribing not only the frequency and the amplitude but also the phase of oscillations. This is bound to create a lot of unnecessary effort trying to shift the phase when it is offset, something a swinger does not really care about.

\subsection{Naive limit cycle} 

The next idea is to use the standard construction of \noindent\textbf{\textit{imposing a  stable limit cycle}} on a single second-order equation \cite{Hak12,Iwa}. One specifies the desired limit cycle of the form $\pfi(\phi,\dphi)=r^2$, where $\pfi$ is an \noindent\textbf{\textit{``energy function"}} determining its shape, e.g. $\pfi=\frac12\dphi^2+\frac12\o^2\phi^2$ or $\pfi=\frac12\dphi^2+\frac12\o^2(1-\cos\phi)$, and $r$ determines its size, see Figure \ref{PhaseLimCyc}. Then $v$ is chosen so that the closed loop equation becomes
\begin{equation}\label{LimCyc}
\ddphi+\kappa(\pfi-r^2)\,\frac{\partial\pfi}{\partial\dphi}+\frac{\dphi}{\partial\pfi/\partial\dphi}\,\frac{\partial\pfi}{\partial\phi}=0,
\end{equation}  
with a positive gain $\kappa$. One can then prove, using the Lyapunov function $V=\frac12(\pfi-r^2)^2$ and La Salle's invariance principle, that $\pfi=r^2$ is its stable limit cycle \cite{AIb}. For the pendulum $\pfi=\frac12\dphi^2+\frac12\o^2(1-\cos\phi)$ is the most natural choice, and it leads to \begin{equation}\label{Stabphi}
\ddphi+\kappa(\pfi-r^2)\dphi+\o^2\sin\phi=0
\end{equation}  
with a very transparent interpretation. When $\pfi>r^2$ the damping term is positive and decreases the amplitude, and when $\pfi<r^2$ it turns negative and the energy is pumped into the swing instead. Asymptotically, $\pfi\approx r^2$ and we have oscillations at the natural frequency as with the naive control, but now they are stabilized against perturbations.
\begin{figure}[!ht]
\centering
(a) \includegraphics[width=2.7in]{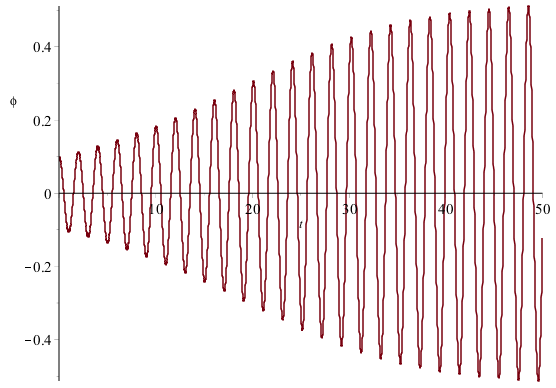} 
(b) \includegraphics[width=2.7in]{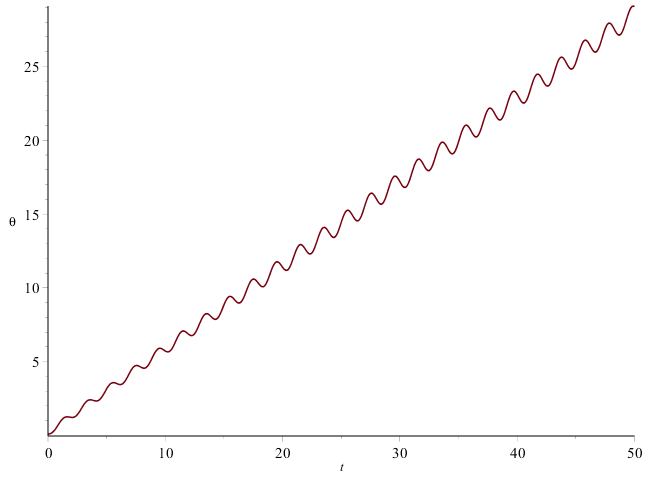}
\caption{\label{Naive0D} Graphs of (a) $\phi(t)$; (b) $\theta(t)$ for balanced swing with the limit cycle control  on $\phi$: $M=1$, $I_1=1$, $I_2=0.04$, $\alpha=0.1$, $\kappa=0.1$, $r=1.16$.
} 
\end{figure} 

While this is a fine solution for stabilizing oscillations of a damped pendulum, applied directly to \eqref{psiBalSw} it runs into a problem well-known from stabilizing equilibria in underactuated systems. When the controlled subsystem is close to the equilibrium point or manifold, the remaining degrees of freedom follow what is called \noindent\textbf{\textit{zero dynamics}} \cite[4.3]{Isid}. In this case, we simply ignored $\psi$, and hence $\theta$. But in zero dynamics the stabilizing control reduces to the damping cancellation, and the second equation from \eqref{BalSw} becomes
$\ddth=-\ddphi-\frac{\alpha}{I_2}\dphi$. Double integration then produces a secular term that grows linearly with time. Since $\phi$ undergoes something close to harmonic oscillations on the limit cycle we can expect $\theta$ to grow without a bound for generic initial values, and indeed this is the case, see Figure \ref{Naive0D}. Even when we select initial values to make the secular term vanish, random perturbations and/or parameter imprecision would knock the swing with this control off a bounded trajectory. Since $\theta$ is an angle unbounded growth means that the body will be making full rotations around the pivot point, so this solution is not physically realizable by humans. Even with a robot, this type of control would require a lot of wasted energy to implement. Ruling out full rotations motivated our second design criterion. 

\subsection{Velocity dependent feedback} 

Taking a closer look at human swingers, we see that they typically throw their body back and feet forward on the forward motion of the swing, and reverse this configuration on the backward motion \cite{WRR}, see also Figure \ref{HumSwing}. Spong proposed the same energy pumping strategy, in a feedback form, for swing up control of the Acrobot arm \cite{Spong}. The idea is that the sign of $\theta$ should be matched with the sign of $\dphi$, for example by setting $\theta=\gamma\tan^{-1}\dphi$ with some $\gamma>0$. The problem is that such a control leads to a third order equation for $\phi$, and, although it seems to take care of keeping $\theta$ bounded, it really does not. Once we convert the expression for $\theta$ into one for the corresponding torque we can no longer guarantee that the closed loop solution will be a trajectory where $\theta$ stays bounded. This is not a problem for the swing up task, where the goal is to swing the arm into the vertical upward position and stabilize it there by some other means, e.g. by LQR, but it is for periodic swinging. We need a controller that will itself work somewhat like an LQR, but for a limit cycle orbit rather than for an equilibrium point or a fixed trajectory.

\section{Feedback linearization and imposing a limit cycle}\label{ReaWh}
 
The previous section should suggest to the reader that designing a suitable torque control for the swing is not as simple as it seems, and may require more sophisticated techniques. A more successful approach comes from the observation that the balanced swing system \eqref{BalSw} has the same equations of motion as the reaction wheel pendulum with damping. And with the damping turned off this system is feedback linearizable \cite{SCL,AIb, Griz}. This means that it can be transformed into a chain of integrators by a change of coordinates and of the control input  \cite[4.2]{Isid}\cite[2]{Resp}. The problem of inducing a limit cycle with bounded zero dynamics in feedback linearizable systems has also received some attention in the recent literature \cite{AIb, Hak17}.

Feedback linearization transforms the system into one for new state variables $x_1,y_1$, and the original variables $\phi,\theta$ are expressed as explicit functions of them and their derivatives. This is where we get lucky the second time -- these functions turn out to be such that imposing linear oscillatory dynamics on $x_1,y_1$ has the effect of driving $\phi$ to a limit cycle with bounded zero dynamics for $\theta$. To introduce the new variables, note that the system \eqref{psiBalSw} with $c=0$ is just
\begin{align}\label{psiBalSwc0}
\begin{cases}
\ddphi+\o^2\sin\phi=v\\
\ddpsi=v.
\end{cases}
\end{align}
It is feedback linearized by the output $x_1$ below, and we define the corresponding state vector $x$ and the new control variable $w$  as follows:
\begin{align}\label{RWLin}
\begin{split}
x_1&:=\frac1{\o^2}(\psi-\phi)\\
x_2&:=\dot{x}_1=\frac1{\o^2}(\dpsi-\dphi)\\
x_3&:=\dot{x}_2=\sin\phi\\
x_4&:=\dot{x}_3=\dphi\cos\phi\\
w&:=\dot{x}_4=\ddphi\cos\phi-\dphi^2\sin\phi=-(\dphi^2+\o^2\cos\phi)\sin\phi+\cos\phi\,v.
\end{split}
\end{align}
\begin{remark}Note that one can solve for $v$ given $w$ only if $\cos\phi$ does not vanish, so this linearizing transformation is only well-defined in the range $|\phi|<\frac{\pi}2$.
\end{remark}
In terms of $x$ the system is transformed into the Brunovsky canonical form for $n=4$:
\begin{equation}\label{Brun}
\dot{x}=A_nx+B_nw,\hspace{0.5em} \text{ where } 
A_n:=\begin{pmatrix}
0  &  1  &  0   & \dots  &  0 \\
0  &  0  &  1   & \dots  &  0\\
\vdots  &  \vdots  & \vdots  &  \ddots  &  \vdots\\
0  &  0  &  0  &   \dots  &  1\\
0  &  0  &  0  &   \dots  &  0
\end{pmatrix}
\text{ and }
B_n:=\begin{pmatrix}
0\\
0\\
\vdots\\
0\\
1\end{pmatrix}.
\end{equation}
General approach to inducing stable limit cycles on a two-dimensional subspace of such a system, with a bounded oscillatory dynamics on its complement, is described in \cite{AIb,Hak17}. We will simplify the dynamic state feedback construction of \cite{AIb} to a static state feedback form. The idea is to impose a stable limit cycle on $y=Cx$ by applying the two-dimensional construction \eqref{LimCyc}, where $y=(y_1,y_2)^T$ is a vector and $C$ is a $2\times n$ matrix, while ensuring stable zero dynamics for the uncontrolled degrees of freedom. 
\begin{remark} To apply the two-dimensional construction we need to isolate an autonomous two-dimensional system for the dynamics of $y$. Assuming, without loss of generality, that it is also in the Brunovsky canonical form, $C$ has to be selected so that $CA_n$ and $CB_n$ reduce to $A_2$ and $B_2$. Playing with the system's equations leads to imposing the following identities: 
\begin{equation}\label{CId}
CA_n=A_2C+B_2\lambda^T, \ \ \ \ 
CB_n=B_2,
\end{equation} 
where $\lambda$ is some gain vector that provides flexibility for constraining the zero dynamics. The first identity implies that the vector must be of the form $\lambda^T=(0,0,\lambda_1,\dots,\lambda_{n-2})$, and the second then implies that \begin{equation}\label{Cmat}
C=\begin{pmatrix}
\lambda_1  &  \lambda_2  &  \dots & 1 &  0 \\
0  &  \lambda_1  &  \dots  &  \lambda_{n-2} & 1\\
\end{pmatrix}.
\end{equation}
Conversely, $C$ of the form \eqref{Cmat} satisfies identities \eqref{CId}.
\end{remark}
\begin{theorem}\label{nDLimCyc} For the chain of integrators \eqref{Brun} set $y:=Cx$ with $\lambda_i$ selected so that the polynomial $p(s)=s^{n-2}+\sum_{i=1}^{n-2}\lambda_is^{i-1}$ is Hurwitz, i.e. its roots have negative real parts. Then the control 
\begin{equation}\label{yLimCyc}
w=-\O^2y_1-\kappa(\pfi-r^2)\dot{y}_1
-\sum_{i=1}^{n-2}\lambda_ix_{i+2}\  \text{ with }\  \pfi=\frac12\dot{y}_1^2+\frac12\O^2y_1^2\  \text{ and }\  \kappa,r>0
\end{equation}
produces a closed loop system where $y$ has a stable limit cycle of harmonic oscillations with the frequency $\O$ and the amplitude $\frac{r\sqrt{2}}{\O}$, while the complementary state variables $x_1,\dots,x_{n-2}$ have bounded oscillatory zero dynamics.
\end{theorem} 
\begin{proof} 
Since $y=Cx$ we have $\ds{y_1=x_{n-1}+\sum_{i=1}^{n-2}\lambda_ix_{i}}$. Applying equation \eqref{Brun} and identities \eqref{CId}, we obtain 
$$
\dot{y}=C\dot{x}=CA_nx+CB_nw=A_2Cx+B_2\lambda^Tx+B_2w=A_2y+B_2(w+\lambda^Tx).
$$
Therefore, $\ds{y_2=\dot{y}_1=x_{n}+\sum_{i=1}^{n-2}\lambda_ix_{i+1}}$, and
$$
\ddot{y}_1=\dot{y}_2=w+\sum_{i=1}^{n-2}\lambda_i\ddot{x}_i=w+\sum_{i=1}^{n-2}\lambda_ix_{i+2}.
$$
Applying the strategy used in \eqref{LimCyc} with $\pfi=\frac12\dot{y_1}^2+\frac12\O^2y_1^2$ we define the control by \eqref{yLimCyc} and
obtain a closed loop system with a stable limit cycle for $y_1,y_2$. Note that the values of $\lambda_i$ are left undetermined so far.

To complete the proof we have to account for the zero dynamics of the remaining degrees of freedom, this will also provide the constraints on the choice of $\lambda_i$. The initial variables $x_1\dots,x_{n-2}$ together with $y_1,y_2$ can be chosen as coordinates on the entire state space, and we define the complementary subspace state vector as $z:=(x_1,\dots,x_{n-2})^T$. Since $\dot{x}_{i}=\dot{x}_{i+1}$ with $\dot{x}_{n-2}=x_{n-1}=y_1-\sum_{i=1}^{n-2}\lambda_ix_{i}$ we can express $\dot{z}$ in the matrix form as 
\begin{equation}\label{zDynam}
\dot{z}=\begin{pmatrix}
x_2\\
x_3\\
\vdots\\
x_{n-2}\\
y_1-\sum_{i=1}^{n-2}\lambda_ix_{i}\end{pmatrix}
=M_{\lambda}z+B_{n-2}y_1, \text{ where } M_{\lambda}:=
\begin{pmatrix}
0  &  1  &  0   & \dots  &  0 \\
0  &  0  &  1   & \dots  &  0\\
\vdots  &  \vdots  & \vdots  &  \ddots  &  \vdots\\
0  &  0  &  0  &   \dots  &  1\\
-\lambda_1  &  -\lambda_2  &  -\lambda_3  &   \dots  &  -\lambda_{n-2}
\end{pmatrix}.
\end{equation}
The matrix $M_{\lambda}$ is the companion matrix of the polynomial $p(s)=s^{n-2}+\sum_{i=1}^{n-2}\lambda_is^{i-1}$, and its eigenvalues are its roots \cite[4.11]{Lanc}. If $p(s)$ is Hurwitz then the zero dynamics of $z$ is driven by a harmonic excitation $y_1$ with a Hurwitz matrix $M_{\lambda}$, i.e. it is bounded and oscillatory.
\end{proof}  
\begin{remark} One can see this result as generalizing the static state feedback construction of imposing a limit cycle on a single second order equation discussed in Section ``Playground swing: equations of motion and naive controls", and it is inspired by the dynamic state feedback construction of \cite{AIb}. Hurwitz polynomials are polynomials all of whose roots have negative real parts. They commonly appear in control theory because when the characteristic polynomial of the matrix of a linear system is Hurwitz the system is asymptotically stable. In the operational calculus they appear in the denominator of the transfer function, and their roots are its poles. Accordingly, a popular method of linear control design is known as \textbf{\textit{pole placement}}. There are well-known computational methods for testing whether a polynomial is Hurwitz, such as the Hermite-Biehler test and the Routh-Hurwitz criterion. 

To generate a Hurwitz polynomial $p(s)$ one can directly assign its roots, and recover the polynomial's coefficients by the Vieta's formulas. For the system in Theorem \ref{nDLimCyc} this amounts to pole placement for the transfer function of the zero dynamics. More sophisticated generation methods, based on backward application of the Hermite-Biehler test or the Routh-Hurwitz criterion, are discussed  in \cite{ShCh}. They are particularly useful when sampling polynomials to obtain a more desirable behavior by trial and error..
\end{remark}
In our case, $\lambda$ has only two gains $\lambda_1,\lambda_2$, so $y_1=x_3+\lambda_1 x_1+\lambda_2 x_2$, and $x_3=\sin\phi$. Hence, explicitly:
\begin{align}\label{RWy1w}
\begin{split}
y_1&:=\sin\phi+\lambda_1 x_1+\lambda_2\,\dot{x}_1\\
w&=-\O^2y_1-\kappa(\pfi-r^2)\dot{y}_1-\lambda_1\sin\phi-\lambda_2\,\dphi\cos\phi\\
v&=\frac1{\cos\phi}\left(w-(\dphi^2+\o^2\cos\phi)\sin\phi\right).
\end{split}
\end{align}
Naively, we would prefer to select $\lambda_i=0$, so that the limit cycle is imposed directly on $\phi$. This is what we tried with the canceling control in the previous section, and now it becomes clear why it did not work. The polynomial  $p(s)=s^2+\lambda_2s+\lambda_1$ is Hurwitz if and only if $\lambda_i>0$. To get bounded zero dynamics $\theta,\dphi$ and $\dth$ have to be mixed in into the variable upon which the limit cycle is imposed. The angle of the swing would still undergo bounded oscillations under this refined construction because the entire state vector does, although we do lose some control over their exact shape.

\section{Feedback linearization control: analysis and variations}\label{ReaWhVar}

In this section we take a closer look at the swing dynamics under the control law designed in the previous section, and discuss some variations on it. In terms of $x_1$, $y_1$, the closed loop system is the following:
\begin{align}\label{RWClosed}
\begin{cases}
\ddot{x}_1+\lambda_2\dot{x}_1+\lambda_1 x_1=y_1\\
\ddot{y}_1+\kappa(\pfi-r^2)\dot{y}_1+\O^2y_1=0.
\end{cases}
\end{align}  
The second equation imposes a stable limit cycle on $y_1$, and in the steady state $y_1$ oscillates harmonically with the frequency $\O$ and the amplitude $\frac{r\sqrt{2}}{\O}$. The first equation then describes a damped harmonic oscillator driven by $y_1$. The original angles are recovered by inverting the linearizing transformation \eqref{RWLin}:
\begin{align}\label{RWInv}
\begin{split}
\phi&=\sin^{-1}\left(\ddot{x}_1\right);\\
\theta&=-\left(1+\frac{I_1}{I_2}\right)\phi-\frac{Mgl_1}{I_2}x_1.
\end{split}
\end{align}
\begin{figure}[!ht]
\centering
\hspace{-4.0em}
\includegraphics[width=0.9in]{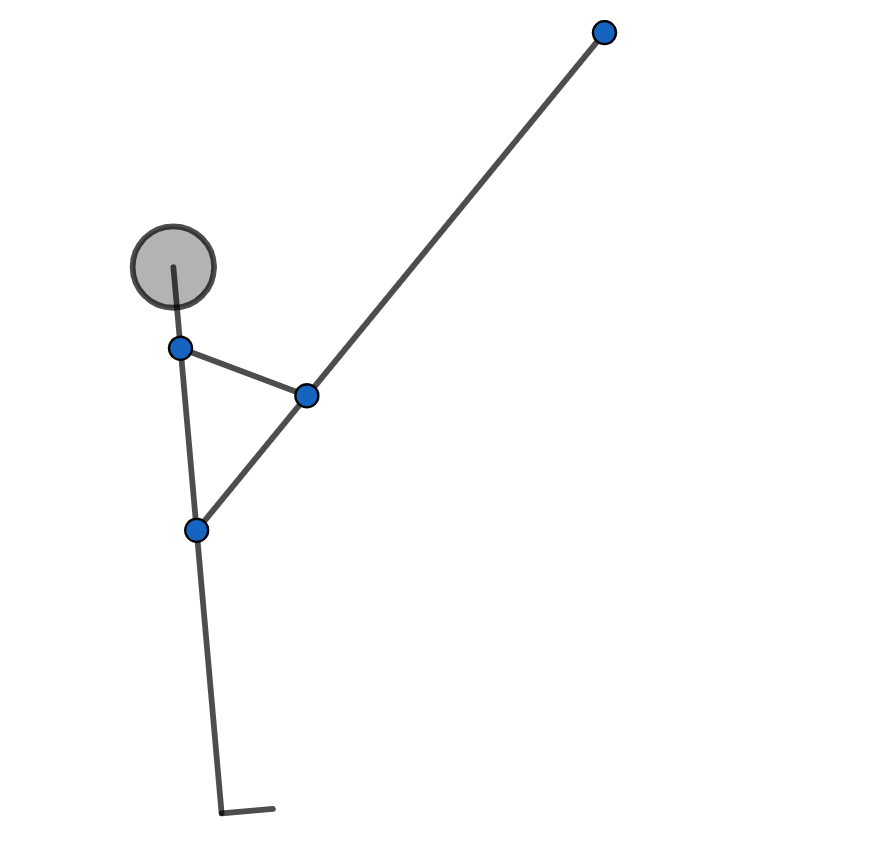}\hspace{10.5em} 
\includegraphics[width=0.85in]{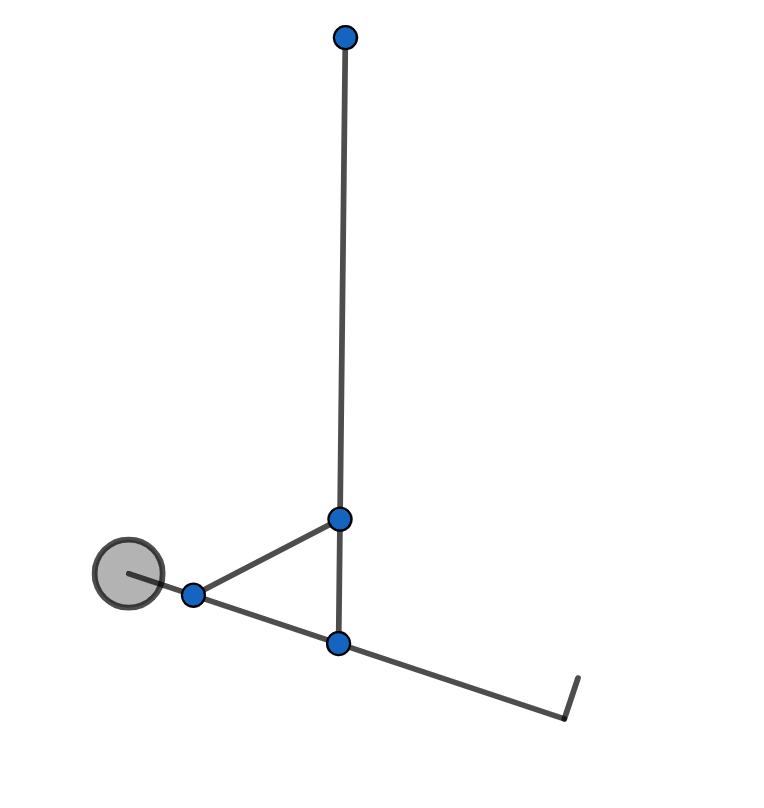}\hspace{4em}
\includegraphics[width=1.1in]{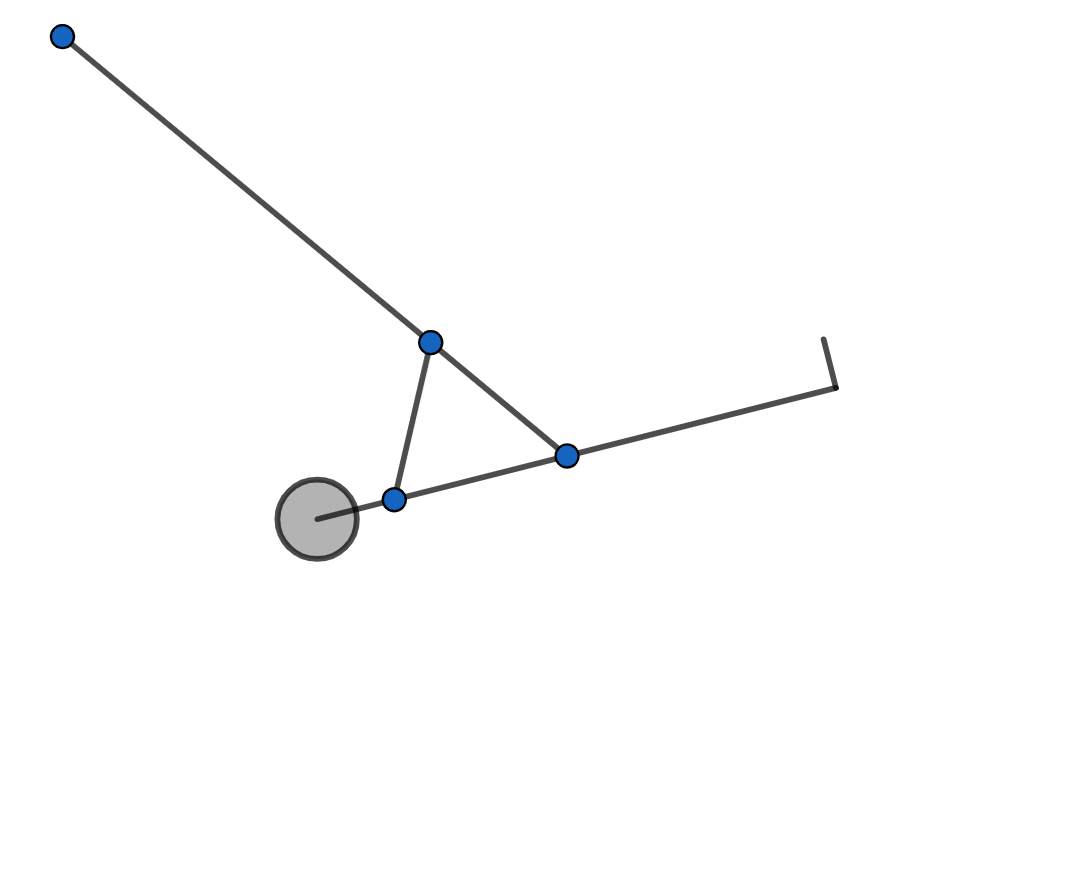}\\
\hspace{-6em}\includegraphics[width=1.24in]{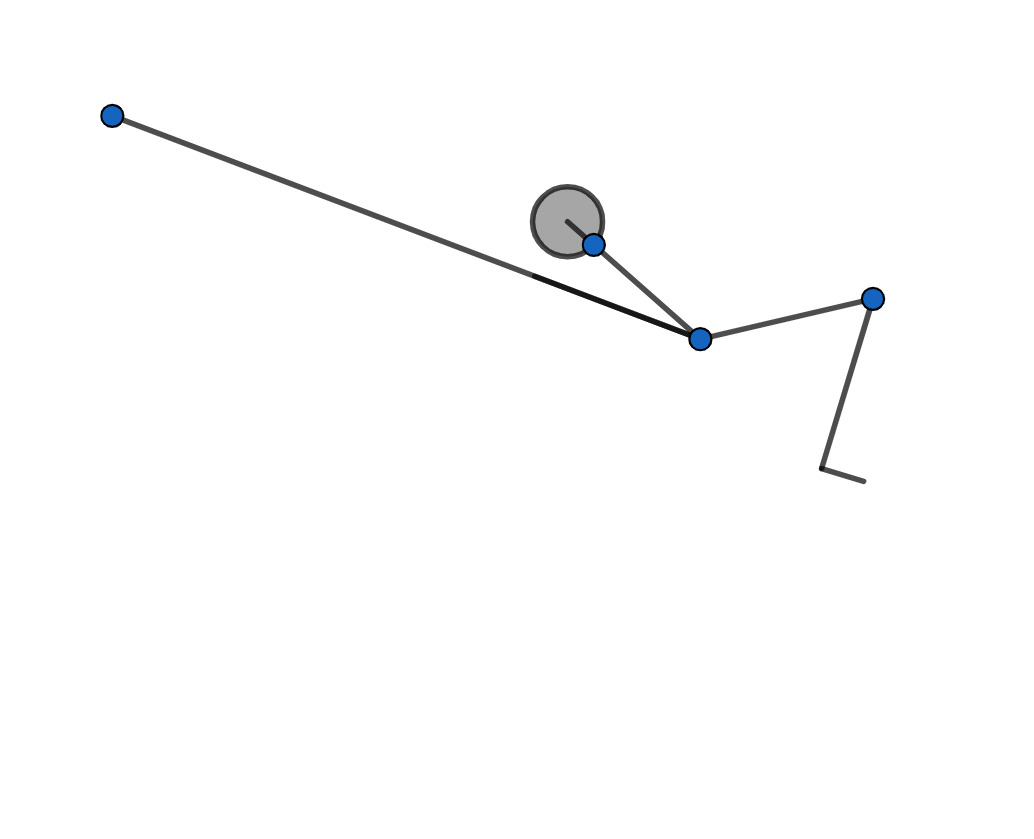}\hspace{5em}
\includegraphics[width=0.86in]{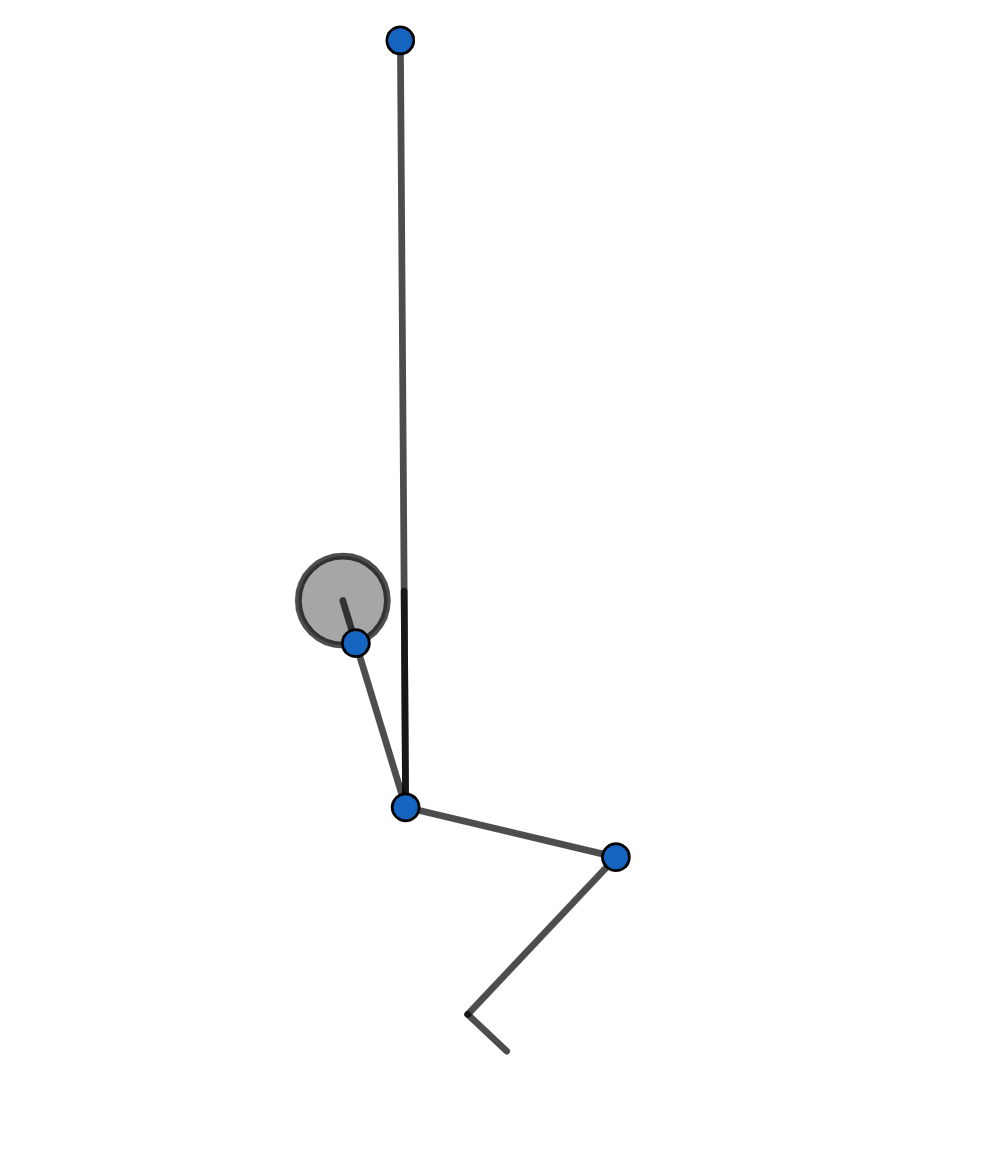}\hspace{0em}
\includegraphics[width=0.82in]{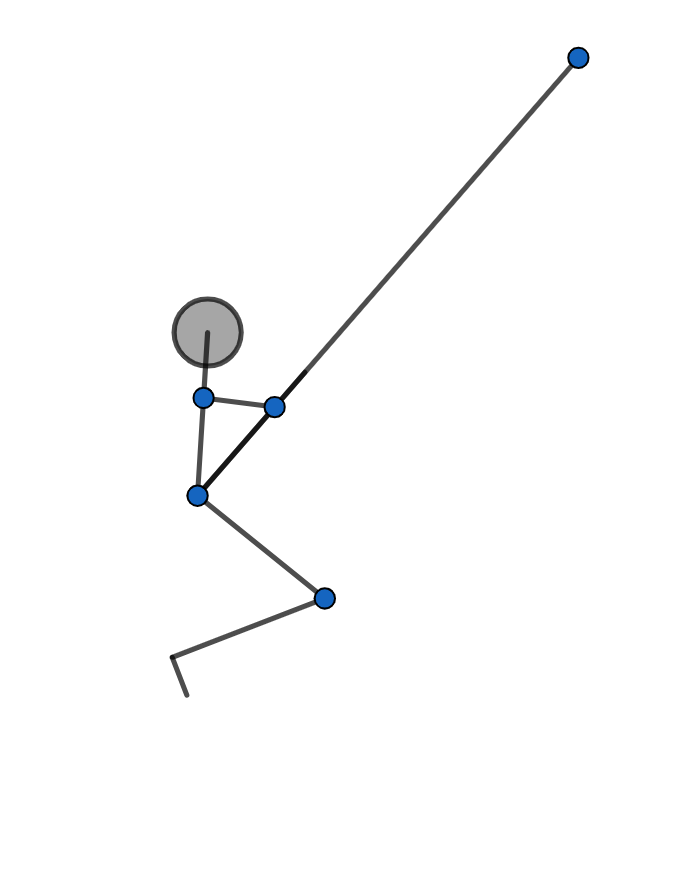}
\caption{\label{HumSwing}Sketches of a human swinger (left to right, top to bottom) made from playground observations.} 
\end{figure}
It is the special form of these relations, in addition to the properties of the feedback linearized system, that is responsible for the resulting control being suitable for our purposes. They imply that $\phi$ and $\theta$ undergo bounded anharmonic oscillations with a single dominant frequency $\O$ in the steady state. Using the standard results on driven harmonic oscillators \cite[9.10.3]{Dav}, in conjunction with \eqref{RWInv}, gives the following theorem.
\begin{theorem}\label{RWHosc} Let $\kappa,r, \lambda_i>0$ and $\pfi=\frac12\dot{y}_1^2+\frac12\O^2y_1^2$. Then the closed loop system for the undamped balanced swing \eqref{BalSw} with the control law \eqref{RWy1w} has bounded oscillatory zero dynamics, and the asymptotic motion of $\phi$ is anharmonic oscillation with the dominant frequency $\O$ and the amplitude\\ ${\ds \phi_{\max}=\sin^{-1}\left(\frac{r\O\sqrt{2}}{\sqrt{(\lambda_1-\O^2)^2+\lambda_2^2\O^2}}\right)}$.
\end{theorem} 
Numerical simulations of the undamped balanced swing \eqref{BalSw} with the control law induced by \eqref{RWy1w} are shown on Figure \ref{RWPendGraphs}\,(a). The physical parameters are selected to match a typical human body and a playground swing, the standard SI units are used throughout the paper. The transient time and the shape of oscillations are quite satisfactory, the amplitudes can be controlled by changing $r$. The anharmonicity is only visible in the shape of the $\theta$ graph. One can also see that the body's and the swing's oscillations are exactly anti-phase in the steady state. This is also a direct consequence of the recovery relations \eqref{RWInv} for this type of control, because $\ddot{x}_1$ and hence $\phi$  oscillates in-phase with $x_1$, while $\theta$ is their linear combination with negative coefficients.
\begin{figure}[!ht]
\centering
(a) \includegraphics[width=2.9in]{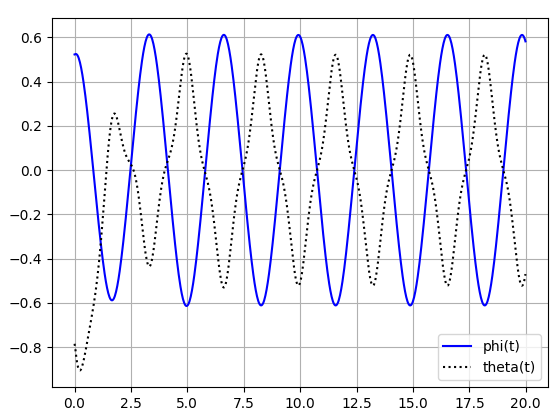} 
(b) \includegraphics[width=2.9in]{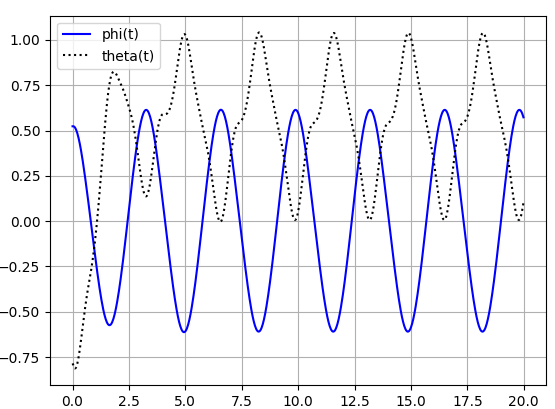}
\caption{\label{RWPendGraphs} Graphs of $\phi(t)$ (solid) and $\theta(t)$ (dotted) for balanced swing with the reaction wheel pendulum control for $\alpha=0$ without shifting the middle angle for $\theta$ in (a), and with shifting it in (b). Other parameters are (masses in kilograms, lengths in meters) $m_1=1$, $m_2=m_3=35$, $l_1=2.5$, $l_2=l_3=0.9$, $\o=1.98$, $\O=1.9$, $\kappa=100$, $r=0.7$.
} 
\end{figure} 
\begin{remark}The type of swinging prescribed by our control does not really resemble human swinging as observed on playgrounds. It is too symmetric, the dumbbell leans one way as much as the other way to transfer as much momentum as possible. Humans are constrained in how much they can lean backwards without falling over, see Figure \ref{HumSwing}, and some swings have benches with back rests that prevent negative values of $\theta$ altogether. 
\end{remark}
These constraints may or may not be relevant for an electromechanical swinging device or a robot, but, in any case, they can be taken into account even by the type of control we proposed, with a small modification. It turns out that all we have to do is shift the center of the limit cycle for $y_1$. Changing the energy function $\pfi$ from Theorem \ref{RWHosc} to $\pfi_\s:=\frac12\dot{y}_1^2+\frac12\O^2(y_1-\s)^2$ has the effect of shifting the center of oscillations for $y_1$ from $0$ to $\s$, and for $x_1$ from $0$ to $\frac{\s}{\lambda_1}$. This does not affect $\phi$ because it is related to $x_1$ through the second derivative, but for $\theta$ the center is shifted to $-\frac{Mgl_1}{I_2}\frac{\s}{\lambda_1}$. One can adjust $r$ and $\s$ so that $\theta$ oscillates between $0$ (or some small negative minimum) and a positive maximum. Numerical simulations of the system with the shift that makes $\theta=0$ the minimal value (no leaning backwards) are shown on Figure \ref{RWPendGraphs}\,(b). 
\begin{figure}[!ht]
\centering
(a)\includegraphics[width=2.94in]{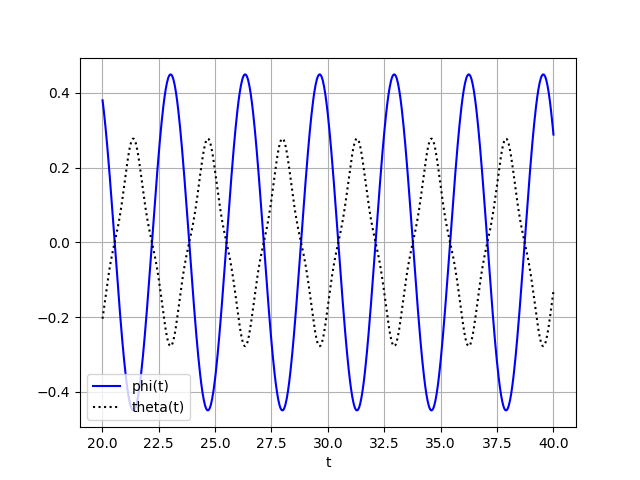}
(b)\includegraphics[width=2.94in]{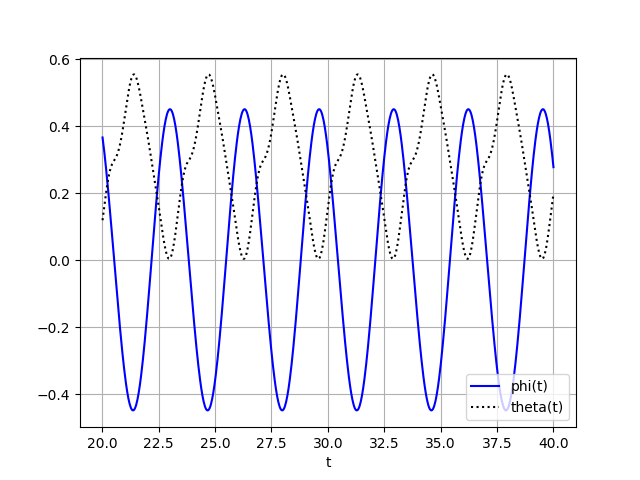}
\caption{\label{RWPendGraphsLow} Graphs of $\phi(t)$ (solid) and $\theta(t)$ (dotted) for balanced swing with the reaction wheel pendulum control for $\alpha=0$ without shifting the middle angle for $\theta$ in (a), and with shifting it in (b). The value of $r$ is reduced to $0.53$ to keep the maximal value of $\theta$ in (b) as in Figure \ref{RWPendGraphs}\,(a). Other parameters are the same as in Figure \ref{RWPendGraphs}.
} 
\end{figure}
\begin{remark}As one can see from Figure \ref{RWPendGraphs}, the shift of the middle angle also changes the shape of body's oscillations, and forces more forward leaning to maintain the same swing amplitude. This may not be comfortable for a human swinger. To remedy this, the amplitude of $\theta$, which is monotonically related to the amplitudes of $\phi$ and $x_1$ by \eqref{RWInv}, should also be reduced by making $r$ smaller. The tradeoff is that the swing amplitude will become smaller as well, see Figure  \ref{RWPendGraphsLow}.
\end{remark}
\begin{figure}[!ht]
\centering
(a) \includegraphics[width=2.85in]{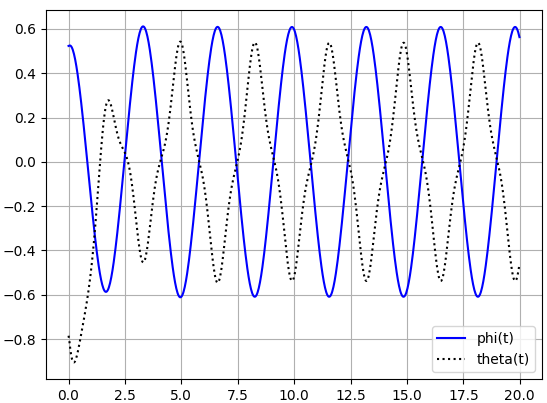} 
(b) \includegraphics[width=2.82in]{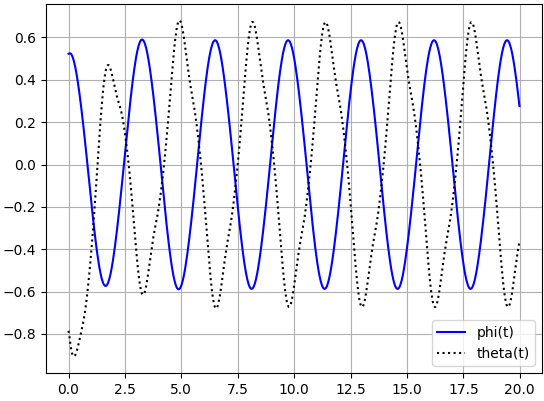}
\caption{\label{RWDPendGraphs} Graphs of $\phi(t)$ (solid) and $\theta(t)$ (dotted) for balanced swing with the reaction wheel pendulum control for (a) $\alpha=3$ and (b) $\alpha=30$. 
Other parameters are $m_1=1$, $m_2=m_3=35$, $l_1=2.5$, $l_2=l_3=0.9$, $\o=1.98$, $\O=1.9$, $\kappa=100$, $r=0.7$.} 
\end{figure} 
Even though the feedback linearization control was designed for the case of zero damping it works quite well when the damping is turned on as well, with small corrections. The new expression for $\ddot{x}_1$ becomes  $\ddot{x}_1=\sin\phi+\frac{c}{\o^2}\dphi$ and we have to modify the control $v$ in \eqref{RWy1w} accordingly: 
\begin{equation}
v=\frac1{\cos\phi+\frac{c}{\o^2}\dphi}\left(w-(\dphi^2+\o^2\cos\phi)\sin\phi\right).
\end{equation}
With these corrections, the closed loop system becomes 
\begin{align}\label{RWDClosed}
\begin{cases}
\ddot{x}_1+\lambda_2\dot{x}_1+\lambda_1 x_1=y_1+\frac{c}{\o^2}\dphi\\
\ddot{y}_1+\kappa(\pfi-r^2)\dot{y}_1+\O^2y_1=-c\left(\lambda_2\sin\phi+\dphi\cos\phi-\frac{\lambda_1-\lambda_2c}{\o^2}\,\dphi\right).
\end{cases}
\end{align} 
Recall that $c=\frac{\alpha}{I_1}$ and the value of $I_1$ is quite large for human swingers, so $c$ functions as a small parameter, partially explaining why the control still works. However, system \eqref{RWDClosed} is analytically opaque, especially considering that there is no simple expression for $\phi$ in terms of $x_1$ and $y_1$ anymore. 

Numerical simulations of the balanced swing \eqref{BalSw} with relatively small damping in the upper joint are shown on Figure \ref{RWDPendGraphs} (compare to sketches on Figure \ref{HumSwing}). One can see that the amplitude of the body needed to achieve the same amplitude of the swing increases with the increase of damping, to compensate for the additional energy loss, and the oscillations of the swing and body angles shift away from being exactly anti-phase at zero damping, which does match playground observations better. An animation of the corresponding swing motion, with a shift that allows some backward leaning, is shown on Figure \ref{RWDPendAnim}. As one would expect, the main difference from human swinging is that the rotation of the body is spread out over the entire swing cycle rather than confined to swift motions near the turning points.

\begin{figure}[!ht]
\centering
\includegraphics[width=1.9in]{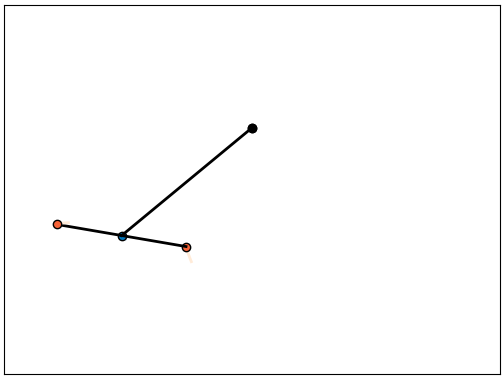} 
\includegraphics[width=1.9in]{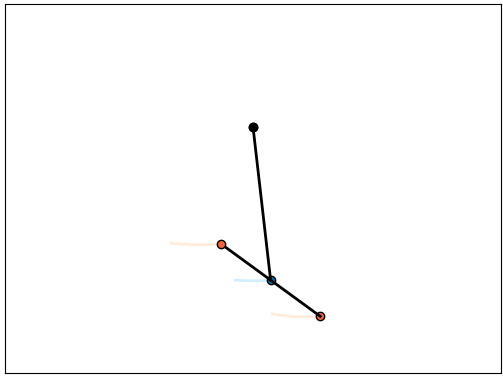}
\includegraphics[width=1.9in]{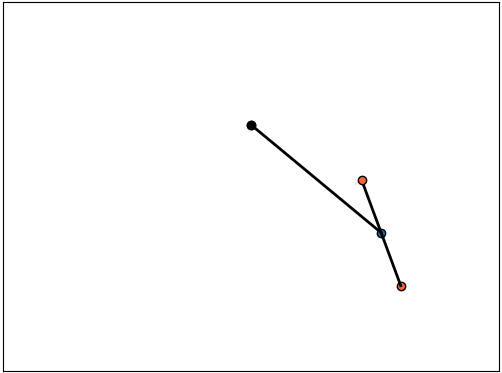}\\
\includegraphics[width=1.9in]{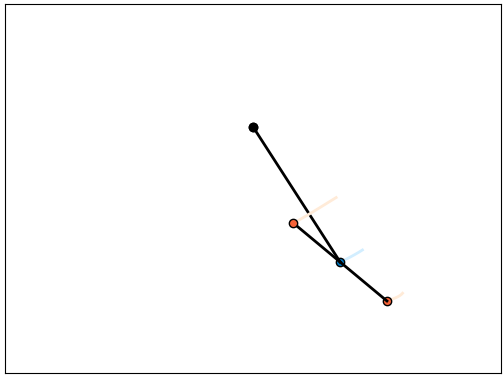} 
\includegraphics[width=1.9in]{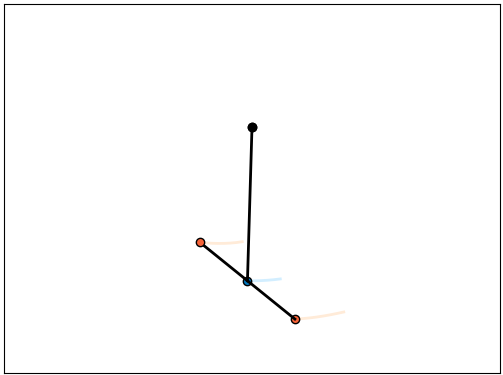}
\includegraphics[width=1.9in]{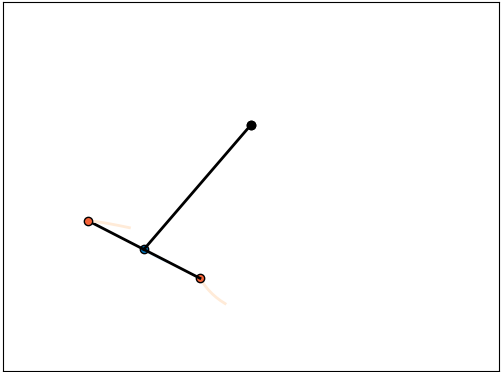}
\caption{\label{RWDPendAnim}Comic strip of a Python animation (left to right, top to bottom) of the balanced swing with $\alpha=3$ controlled by the feedback linearization control with shifted middle angle.} 
\end{figure}

\section{Damping corrections for small amplitudes}\label{Lin}

While the damped ($c\neq0$) reaction wheel pendulum system \eqref{psiBalSw} is not feedback linearizable we can get some insight into its behavior by considering the case of small amplitudes, when $\sin\phi\approx\phi$. The linearized system is analytically tractable, and we can better understand what happens in the closed loop system when a similar control is applied, unlike for the perturbed system \eqref{RWDClosed}. As a bonus, it will suggest an alternative control law for the original non-linear system that is applicable when the damping is large. 

The linearization of system \eqref{psiBalSw} for small $\phi$ is
\begin{align}\label{LinBalSw}
\begin{cases}
\ddphi+c\dphi+\o^2\phi=v\\
\ddpsi=v,
\end{cases}
\end{align} 
and can be written in the matrix form by taking $\xi_1:=\psi$, $\xi_2:=\dpsi$, $\xi_3:=\phi$, and $\xi_4:=\dphi$:
\begin{equation*}\label{LinBalSwMat}
\dot{\xi}=A\xi+Bv,\hspace{0.5em} \text{ where } 
A=\begin{pmatrix}
0  &  1  &  0  &  0\\
0  &  0  &  0  &  0\\
0  &  0  &  0  &  1\\
0  &  0  &  -\o^2  & -c
\end{pmatrix}
\text{ and }
B=\begin{pmatrix}
0\\
1\\
0\\
1\end{pmatrix}.
\end{equation*} 
An output that converts it into the Brunovsky canonical form \eqref{Brun} can be chosen as $h\cdot\xi$ \cite[4.1]{Isid}, \cite{Resp}, where $h$ is a vector orthogonal to $B$, $AB$, $A^2B$. When we choose $x_1$ as below the conversion runs parallel to \eqref{RWLin}, but with explicit damping corrections:
\begin{align}\label{LinVar}
\begin{split}
x_1&:=\frac1{\o^2}(\psi-\phi)-\frac{c}{\o^4}(\dpsi-\dphi)+\frac{c^2}{\o^4}\phi\\
x_2&:=\dot{x}_1=\frac1{\o^2}(\dpsi-\dphi)+\frac{c}{\o^2}\dphi\\
x_3&:=\dot{x}_2=\phi\\
x_4&:=\dot{x}_3=\dphi\,.
\end{split}
\end{align}
Indeed, for $c=0$ the expressions for $x_1$ and $x_2$ are exactly as before, and $x_3$ is $\phi$ instead of $\sin\phi$, which are close for small amplitudes. Effectively, the new $x_1$ is the approximate second anti-derivative of $\sin\phi$ for $c\neq0$, in explicit feedback form. The auxiliary control $w:=-\o^2\phi-c\dphi+v$ that transforms 
\eqref{LinBalSw} into the Brunovsky form \eqref{Brun}, and the auxiliary variable $y_1:=\lambda_1x_1+\lambda_2x_2+x_3$, upon which the limit cycle is imposed according to Theorem \ref{nDLimCyc}, also have similar form. The resulting control law is
\begin{equation}\label{LinCont}
v=w+\o^2\phi+c\dphi
=(\o^2-\lambda_1)\phi+(c-\lambda_2)\dphi
-\O^2y_1
-\kappa(\pfi-r^2)\dot{y}_1,
\end{equation} 
and the closed loop system in terms of $x_1$, $y_1$ takes the exact same form as \eqref{RWClosed}:
\begin{align}\label{RWClosedLin}
\begin{cases}
\ddot{x}_1+\lambda_2\dot{x}_1+\lambda_1 x_1=y_1\\
\ddot{y}_1+\kappa(\pfi-r^2)\dot{y}_1+\O^2y_1=0.
\end{cases}
\end{align} 

However, because we now also have $\ddot{x}_1=\phi$, and the equations are linear, we can rewrite \eqref{RWClosedLin} in terms of $\phi$ directly. Differentiating the first equation twice eliminates $x_1$, and solving for $\ddot{y}_1$ in the second equation we derive an alternative form of the closed loop system:
\begin{align}\label{LinClos}
\begin{cases}
\ddphi+\lambda_2\dphi+\lambda_1\phi=-\O^2y_1-\kappa(\pfi-r^2)\dot{y}_1\\
\ddot{y}_1+\kappa(\pfi-r^2)\dot{y}_1+\O^2y_1=0.
\end{cases}
\end{align}
The second equation again produces a stable limit cycle consisting of harmonic oscillations with the frequency $\O$ and the amplitude $\frac{r\sqrt{2}}{\O}$. The first equation is a damped harmonic oscillator driven, in the steady state limit, by a harmonic excitation $-\O^2y_1$. 
\begin{remark} 
The closed loop system is particularly transparent when choosing the natural gains $\lambda_1=\o^2$ and $\lambda_2=c$, as the expression for $v$ suggests. Then $y_1=\psi$ and  $v=-\O^2\psi-\kappa(\pfi-r^2)\dpsi$. In other words, we are imposing a prescribed limit cycle on the sum of the swing and body angles $\phi+\theta$ (recall that $\psi=-\frac{I_2}{I_1}(\phi+\theta)$) rather than on the swing angle $\phi$ itself. This is somewhat counterintuitive, because it is $\phi$ that we really want to control, but, in hindsight, understandable. The reason the naive law that controlled $\phi$ directly produced unbounded zero dynamics was that the equation for $\psi$ in \eqref{psiBalSw} or \eqref{LinBalSw} has no damping or stiffness of its own. This is exactly what our controls introduce artificially, at the expense of controlling $\phi$ only indirectly. Some of the flexibility in controlling $\phi$ has to be sacrificed for the sake of suppressing unbounded growth of $\psi$, and hence of $\theta$. 
\end{remark}

From the standard results on driven harmonic oscillators \cite[9.10.3]{Dav} we can derive the following. 
\begin{theorem}\label{DDHosc} Let $\kappa,r, \lambda_i>0$ and $\pfi=\frac12\dot{y}_1^2+\frac12\O^2y_1^2$. Then the closed loop system for the linearized balanced swing \eqref{LinBalSw} with the control law \eqref{LinCont} has bounded oscillatory zero dynamics, and the asymptotic motion of $\phi$ is harmonic oscillation with the frequency $\O$ and the amplitude ${\ds \phi_{\max}=\frac{r\O\sqrt{2}}{\sqrt{(\lambda_1-\O^2)^2+\lambda_2^2\O^2}}}$.
\end{theorem} 
This explains why the reaction wheel control law still works (for small amplitudes, at least) with non-zero damping. But if we try to apply linear corrections for $x_1$ from \eqref{LinVar} to the original non-linear system \eqref{psiBalSw} we do not get $\ddot{x}_1$ to come out cleanly as $\sin\phi$, but rather get $\ddot{x}_1=\sin\phi+\frac{c}{\o^2}\dphi(1-\cos\phi)$. The exact second anti-derivative does not exist in a static state feedback form because the system is not feedback linearizable, so there is no exact non-linear analog of $x_1$.

However, things are better for $y_1$. The control law \eqref{LinCont} simplifies already when we take $\lambda_1:\lambda_2=\o^2:c$, i.e. when the control gains are in the same ratio as the natural gains. Moreover, this choice cancels the derivatives in the definition of $y_1$ and reduces it to a linear combination of $\phi$ and $\psi$ only:
$$
y_1=\frac{\l_1}{\o^2}(\psi-\phi)+\phi
=\frac{\l_2}{c}(\psi-\phi)+\phi.
$$
The advantage is that we can define $y_1$ this way even for the non-linear system, and, although we can no longer write the closed loop system in terms of $x_1$ and $y_1$, it turns out that the system \eqref{LinClos} for $\phi$ and $y_1$ does have a non-linear analog, and a well-studied one. This is the approach we pursue in the next section. However, there is a downside. Fixing the gain ratio leaves a single gain parameter to vary and reduces control options. As we will see, such single gain control only works well when the damping in the upper joint of the uncontrolled swing is already large.

\section{From balanced swing to driven damped pendulum}\label{HDDPend}

As we saw for the linearized balanced swing, for a particular choice of the gain ratio the variable upon which the stable limit cycle is imposed reduces to a linear combination of the swing and the body angles, without time derivatives, and the alternative closed loop system \eqref{LinClos} simplifies. We will show that a similar simplification occurs in the original non-linear system, even though, with the damping turned on, it is no longer feedback linearizable. One can see it as implementing the idea that the limit cycle should be imposed on a combination of angles rather than on the swing angle itself to ensure bounded zero dynamics. 

It is convenient to rescale the auxiliary variable $y_1$ we used before and define 
$$
\chi:=-\frac{I_1}{I_2}\psi+\mu\,\phi=\theta+(1+\mu)\phi\,,$$
where $\mu$ is a constant. Then from \eqref{psiBalSw} we obtain
\begin{equation}\label{chiEq}
\ddchi=-\frac{I_1}{I_2}v+\mu\,(v-c\dphi-\o^2\sin\phi)
=-\mu\,(c\dphi+\o^2\sin\phi)+\left(\mu-\frac{I_1}{I_2}\right)v.
\end{equation}
We impose a limit cycle on $\chi$ as in \eqref{LinClos} 
by using the control:
\begin{equation}\label{chiCont}
v=\frac{1}{\mu-\frac{I_1}{I_2}}\Big(\mu\,(c\dphi+\o^2\sin\phi)-\O^2\chi-\kappa(\pfi-r^2)\dchi\Big)\ \text{ with }\ \pfi=\frac12\dchi^2+\frac12\O^2\chi^2,
\end{equation}
and obtain the closed loop system
\begin{align}\label{chiBalSw}
\begin{cases}
\ddphi+\frac{1}{1-\frac{I_2}{I_1}\mu}c\dphi+\frac{1}{1-\frac{I_2}{I_1}\mu}\o^2\sin\phi=\frac{1}{1-\frac{I_2}{I_1}\mu}\frac{I_2}{I_1}\left(\O^2\chi+\kappa(\pfi-r^2)\dchi\right)\\
\ddchi+\kappa(\pfi-r^2)\dchi+\O^2\chi=0.
\end{cases}
\end{align}
We need to select $\mu<\frac{I_1}{I_2}$ for the damping to be positive. In the steady state, $\chi$ is oscillating harmonically with the frequency $\O$ and the amplitude $\frac{r\sqrt{2}}{\O}$, and drives the equation for $\phi$ with the scaled amplitude $\frac{1}{1-\frac{I_2}{I_1}\mu}\frac{I_2}{I_1}r\O\sqrt{2}$. 
\begin{remark} The interpretation of $\chi$ is most transparent when $\mu=0$ and $\chi=\phi+\theta$ (this corresponds to choosing the natural gains in the linearized control of the previous section), and when $\mu=-1$ and $\chi=\theta$. Moreover, taking $\mu\to-\infty$ corresponds to imposing the limit cycle on $\phi$ directly, in which case the two equations in \eqref{chiBalSw} collapse into a single one, and the zero dynamics is unbounded, as we already saw in Section ``Playground swing: equations of motion and naive controls". 
\end{remark}
Thus, in the steady-state limit our problem reduces to driving a damped pendulum by a harmonic excitation
\begin{equation}\label{muDDPend}
\ddphi+\frac{1}{1-\frac{I_2}{I_1}\mu}c\dphi+\frac{1}{1-\frac{I_2}{I_1}\mu}\o^2\sin\phi=\frac{1}{1-\frac{I_2}{I_1}\mu}\frac{I_2}{I_1}r\O\sqrt{2}\cos(\O t+p).
\end{equation}

In the case of small amplitudes $\sin\phi\approx\phi$ we have a driven harmonic oscillator as in the linearized problem \eqref{LinClos}. The steady-state oscillations then have the frequency $\O$ and the amplitude 
\begin{equation}\label{MaxPhiSmall}
{\ds \Phi=\frac{\frac{I_2}{I_1}\,r\O\sqrt{2}}{\sqrt{\left(\o^2-(1-\frac{I_2}{I_1}\mu)\,\O^2\right)^2+c^2\O^2}}}.
\end{equation}
We can select any desired frequency $\O$ and any desired positive damping coefficient by adjusting $\mu$. Then $r$ can be adjusted to set the amplitude $\Phi$. The energy transfer is most efficient at the resonance, when the expression under the square root is minimized, which gives: 
\begin{equation}\label{LinRes}
\O=\frac{\o}{\sqrt{1-\frac{I_2}{I_1}\mu}}\,\sqrt{1-\frac{c^2}{2\o^2\left(1-\frac{I_2}{I_1}\mu\right)}}\,,
\end{equation} 
and this is a natural choice for $\O$, leaving only $r$ and $\mu$ to select.

As before, we can accommodate constraints on the range of $\theta$. It is easiest to do by choosing $\mu=-1$ which gives $\chi=\theta$, we simply choose $\pfi_\s=\frac12\dth^2+\frac12\O^2(\theta-\s)^2$, with some prescribed shifted middle angle $\s$, as the energy function. Then the equation of motion \eqref{LimCyc} for $\theta$ becomes
\begin{equation}\label{Humth}
\ddth+\kappa(\pfi_\s-r^2)\,\dth+\O^2(\theta-\s)=0.
\end{equation}
The amplitude of the steady state oscillations is still given by $\theta_{\max}=\frac{r\sqrt{2}}{\O}$, but they are now centered at $\s$ rather than at $0$, and hence the range of angles is $[\s-\theta_{\max},\s+\theta_{\max}]$. 
\begin{figure}[!ht]
\centering
(a) \includegraphics[width=2.9in]{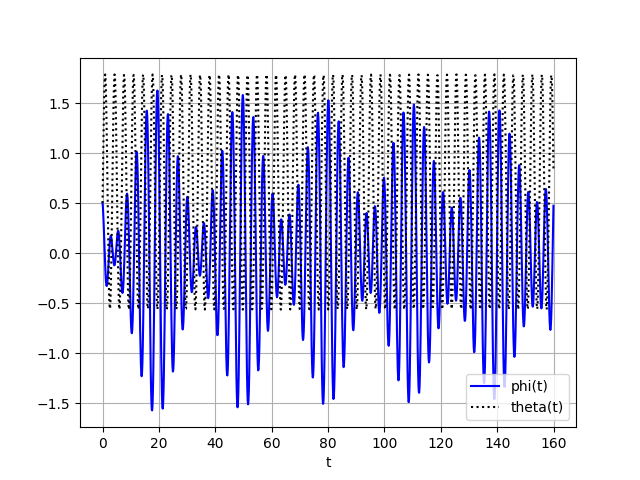} 
(b) \includegraphics[width=2.9in]{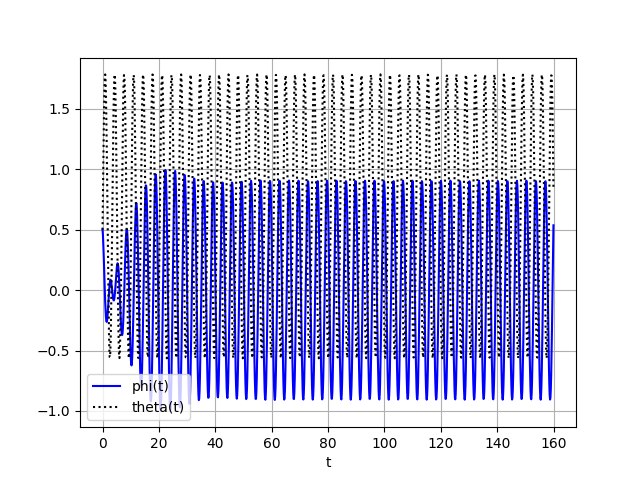}
\\
(c) \includegraphics[width=2.9in]{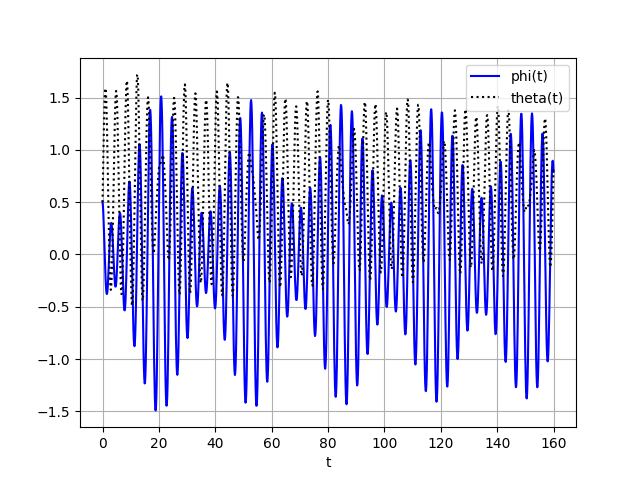} 
(d) \includegraphics[width=2.9in]{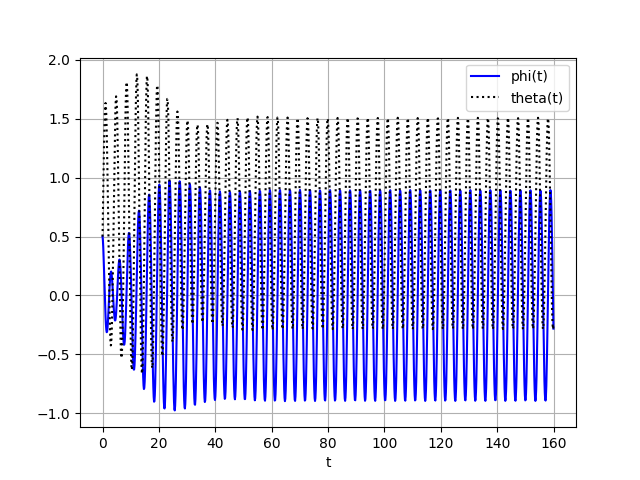}
\caption{\label{DDPendSmall} Graphs of $\phi(t)$ (solid) and $\theta(t)$ (dotted) for balanced swing with the driven damped pendulum control for (a) $\mu=-1$, $\alpha=30$;  (b) $\mu=-1$, $\alpha=100$; (c) $\mu=-2$, $\alpha=30$; (d) $\mu=-2$, $\alpha=100$. Other parameters are $m_1=1$, $m_2=m_3=35$, $l_1=2.5$, $l_2=l_3=0.9$, $\o=1.98$, $\O=1.9$, $\kappa=100$, $r=0.7$. 
} 
\end{figure} 

By choosing $\s=\theta_{\max}=45^\circ$ we get rotations between $0^\circ$ and $90^\circ$ characteristic for swings with back rests, and by choosing $\s=35^\circ$, $\theta_{\max}=55^\circ$ we get rotations between $-20^\circ$ and $90^\circ$ that humans can achieve on swings with flat seats without falling over. Perhaps, one can lean even further back by holding on tight to the swing chains or ropes. The corresponding control is given by \begin{equation}\label{thCont}
v=\frac{1}{1+\frac{I_1}{I_2}}\Big(c\dphi+\o^2\sin\phi+\O^2(\theta-\s)+\kappa(\pfi_\s-r^2)\dth\Big),
\end{equation}
and the zero dynamics equation for the swing angle is 
\begin{equation}\label{shDDPend}
\ddphi+\frac{1}{1+\frac{I_2}{I_1}}c\dphi+\frac{1}{1+\frac{I_2}{I_1}}\o^2\sin\phi=\frac{1}{1+\frac{I_2}{I_1}}\frac{I_2}{I_1}\Big(r\O\sqrt{2}\cos(\O t+p)-\O^2\s\Big).
\end{equation}

Compared to \eqref{shDDPend} (with $\mu=-1$) we have an additional constant driving term $-\frac{1}{1+\frac{I_2}{I_1}}\frac{I_2}{I_1}\,\O^2\s$. Its effect is easy to understand intuitively for small amplitudes when $\sin\phi\approx\phi$. In that case the shift can be subsumed into $\phi$, and we have regular oscillations around a shifted middle position $-\frac{I_2}{I_1}\,\frac{\O^2}{\o^2}\s$. Because typically $\O\approx\o$ and $I_2\ll I_1$ the shift in $\phi$ is barely noticeable in practice. This approximation remains valid roughly for $\phi_{\max}\lesssim 0.4$ or $23^\circ$.

Damped pendulum driven by a harmonic excitation is one of the simplest, and best studied, mechanical systems that can exhibit chaotic behavior \cite{BB,BG}. For small driving amplitudes and frequencies it is well approximated by driven damped harmonic oscillator, which goes into resonance when the driving and the natural frequencies are close to each other. The dynamics remains qualitatively similar even as the approximation is no longer accurate for somewhat larger values. But this does not remain the case forever. As the driving frequency or amplitude are increased, the pendulum starts behaving differently in a dramatic fashion, its motion undergoes a cascade of period doubling bifurcations. Period doubling means that the oscillations no longer have a single dominant frequency. It splits in two, and the graph of the pendulum angle as a function of time looks like a wave of one frequency modulated by another, the pattern known as \textbf{\textit{beats}}, see Figure \ref{DDPendSmall}\,(a),\,(c). As more and more frequencies are superimposed, the motion becomes increasingly erratic. The period doubling cascade ends in fully chaotic behavior, which, surprisingly, gives way back to regular oscillations at still higher values of the frequency or the amplitude, but only for narrow ranges of their values \cite[3.3.5]{BG}. 

Our simulations show that for realistic values of lengths and masses the closed loop driven damped pendulum falls within the range of regular oscillations (and hence satisfies our design criteria) only for relatively large values of the damping coefficient $\alpha$, see Figure \ref{DDPendSmall}. For smaller values of $\alpha$, the steady state oscillations display the pattern of beats, a wave with one frequency modulated by another. It appears that the closed loop system undergoes a period doubling bifurcation as $\alpha$ is decreased, past which this type of control is no longer satisfactory. 
\begin{remark}
The appearance of beats is a price of having a more analytically transparent closed loop system \eqref{chiBalSw}. It gives us one less control parameter than the feedback linearization control ansatz. While $\mu$ allows us to increase the damping in the closed loop system \eqref{muDDPend}, it can only do so by also increasing its ``stiffness" (the coefficient of $\sin\phi$) at the same time, which leads to the beats. The only way to have damping large enough without too much spurious oscillation is to have it large enough to begin with, before the control is applied. \end{remark}
To interpret our results, 
we rescale the time variable to $\tau=kt$ with $k^2=\frac{\o^2}{1-\frac{I_2}{I_1}\mu}$ and convert \eqref{muDDPend} into the dimensionless form
\begin{equation}\label{StDDPend}
\ddphi+\d\dphi+\sin\phi=A_D\cos(\O_D t+p), 
\end{equation}
where 
\begin{equation}\label{StPar}
\d:=\frac{c}{\o\sqrt{1-\frac{I_2}{I_1}\mu}};\ \ \ \ \  
A_D:=\frac{I_2}{I_1}\frac{r\O\sqrt{2}}{\o^2};\ \ \ \ \ 
\O_D:=\frac{\O}{\o}\sqrt{1-\frac{I_2}{I_1}\mu}.
\end{equation} 
In terms of the dimensionless parameters the amplitude formula \eqref{MaxPhiSmall} simplifies to
\begin{equation}\label{MaxPhiHarm}
{\ds \Phi=\frac{A_D}{\sqrt{\left(1-\O_D^2\right)^2+\d^2\O_D^2}}},
\end{equation}
Note that the choice $\O_D=1$ corresponds to the resonance in the absence of damping, and \eqref{LinRes} reduces to $\O_D=\sqrt{1-\frac{\d^2}2}$. Regular oscillations, for which this formula is valid, take place for small $A_D$ and $\delta$, with $\O_D$ close to $1$. 

A detailed semi-analytic investigation of the onset of chaos in a harmonically driven damped pendulum by perturbation methods can be found in \cite{Mil}. It is based on Fourier approximation with a dominant mode. It turns out that in the next order of approximation the amplitude formula \eqref{MaxPhiHarm} has to be corrected by replacing $1$ in the denominator with an expression that depends on the amplitude $\Phi$, and so it turns into an implicit equation for $\Phi$. This equation (see \eqref{MaxPhiInp} below) defines curves in the $\Phi$-$\O_D$ plane called the \textbf{\textit{resonance curves}}. They are plotted in \cite{Mil} for some typical parameter values, and have two branches, normal and anomalous. 

There is a range of parameter values where the normal branch represents stable regimes and resembles the resonance curve for the driven damped harmonic oscillator given by \eqref{MaxPhiHarm}. This is the primary range of regular oscillatory behavior. Increasing $A_D$ deforms the resonance curves into multi-valued graphs, and regular oscillations destabilize, setting off the period doubling cascade where the dominant mode approximation is no longer applicable. But it can be used to analytically estimate the threshold values where the onset of the period doubling occurs. We will use the following result from \cite{Mil}.
\begin{theorem}\label{DDPendMil} Suppose that $0\leq \d\leq 0.592$, $\sqrt{1-\frac{\d^2}2}\leq\O_D\leq1$,  and
\begin{equation}\label{ADRange}
A_D\leq3.16\,\d(1-1.16\,\d^2).
\end{equation}
Then the steady state motion of the driven damped pendulum \eqref {StDDPend} is anharmonic oscillation with the dominant frequency $\O_D$ and the amplitude $\Phi$ implicitly determined, to the first order of approximation, by
\begin{equation}\label{MaxPhiInp}
{\ds \Phi=\frac{A_D}{\sqrt{\left(2\frac{J_1(\Phi)}{\Phi}-\O_D^2\right)^2+\d^2\O_D^2}}},
\end{equation}
where $J_1$ is the first Bessel function.
\end{theorem} 
\noindent It is the amplitude condition \eqref{ADRange} that cannot be satisfied when the value of $\alpha$ (and hence of $\delta$) is small. For our choice of swing parameters and $\mu=-2$ the threshold value for falling within the regular oscillatory range is roughly $\alpha\approx40$. Because we only have one control gain $\mu$ to vary, we cannot increase $\delta$ independently of $\O_D$ and $A_D$. Moreover, values of $\mu$ close to the $\frac{I_1}{I_2}$ threshold, that increase $\delta$, lead to numerical instability in simulations. On the other hand, if we keep both gains $\l_1,\l_2$  the asymptotic closed loop dynamics is no longer described by the driven damped pendulum or some other known system, see \eqref{RWDClosed}.

\section{Unbalanced swing}\label{UnBal}

Up to this point we assumed that the swinger is fully balanced, i.e. her center of mass is on the swing's axis. When this is not so, rotating the body (modeled by the attached dumbbell), aside from the momentum transfer, has the additional effect of moving the center of mass up and down. This additional energy pumping mechanism is used exclusively when pumping up the swing from the {\it standing position} \cite{WRR,LaFo}. There is, however, a tradeoff. Pumping the swing up by moving the center of mass is most effective near the midswing, whereas the angular momentum transfer is most effective when it is near the turning points. Since in our dumbbell model both moves can only be accomplished by rotation the two strategies cannot be used simultaneously. 

With the $N$ terms included the equations of motion \eqref{ELDamp} become
\begin{align}\label{UnBalSw}
\begin{cases}
(I_1+Nl_1\cos\theta)\ddphi+(I_2+Nl_1\cos\theta)(\ddphi+\ddth)+Nl_1\sin\theta(\dphi+\dth)^2-Nl_1\sin\theta\dphi^2\\
\hspace{19em} +\alpha\dphi+Mgl_1\sin\phi+Ng\sin(\phi+\theta)=0\\
I_2(\ddphi+\ddth)+Nl_1\cos\theta\ddphi+Nl_1\sin\theta\,\dphi^2+Ng\sin(\phi+\theta)=u-\beta\dth.
\end{cases}
\end{align} 
Considering the prevalence of the angle sum $\phi+\theta$ and its derivatives in the above equations it is natural to introduce $\xi:=\phi+\theta$ as a new variable. In terms of $\xi$, the equations transform into (we keep $\theta=\xi-\phi$ for brevity)
\begin{align}\label{UnBalSwPsi}
\begin{cases}
(I_1+Nl_1\cos\theta)\ddphi+(I_2+Nl_1\cos\theta)\ddxi+Nl_1\sin\theta(\dxi^2-\dphi^2)\\
\hspace{19em} +\alpha\dphi+Mgl_1\sin\phi+Ng\sin\xi=0\\
I_2\ddxi+Nl_1\cos\theta\ddphi+Nl_1\sin\theta\,\dphi^2+Ng\sin\xi=u-\beta\dth.
\end{cases}
\end{align} 
Up to the damping terms and reinterpretation of constants, these are the equations of motion of the Acrobot \cite{Mur,SCL,ZCM}. Subtracting the second equation from the first simplifies the latter to:
\begin{equation}\label{UnBalSwPhi}
I_1\ddphi+\alpha\dphi+Mgl_1\sin\phi+Nl_1(\cos\theta\,\ddxi-\sin\theta\,\dphi^2)=-(u-\beta\dth).
\end{equation}
\begin{remark} We will assume that the imbalance $N=m_2l_2-m_3l_3$ is small, or more precisely, $Nl_1\ll I_1$, but not necessarily $Nl_1\ll I_2$, which is typical for human swingers. This means that $Nl_1$ terms in \eqref{UnBalSwPhi} can be neglected, as they are divided by $I_1$, but should be kept in the second equation in \eqref{UnBalSwPsi}. 
\end{remark}
If we then eliminate $\ddphi$ from it and neglect the $Nl_1/I_1$ term the simplified system becomes
\begin{equation}\label{UnBalSwNeg}
\begin{cases}
I_1\ddphi+\alpha\dphi+Mgl_1\sin\phi\approx-(u-\beta\dth)\\
I_2\ddxi\approx u-\beta\dth-Ng\sin\xi-Nl_1\sin\theta\,\dphi^2.
\end{cases}
\end{equation}
Since $\psi:=-\frac{I_2}{I_1}(\phi+\theta)=-\frac{I_2}{I_1}\xi$ we can rewrite the last system as a perturbation of \eqref{psiBalSw}:
\begin{align}\label{psiUnBalSw}
\begin{cases}
\ddphi+c\dphi+\o^2\sin\phi\approx v\\
\ddpsi\approx v-N\left(\frac{g}{I_1}\sin\left(\frac{I_1}{I_2}\psi\right)-\frac{l_1}{I_1}\sin\left(\frac{I_1}{I_2}\psi+\phi\right)\dphi^2\right).
\end{cases}
\end{align}
\begin{figure}[!ht]
\centering
(a) \includegraphics[width=2.9in]{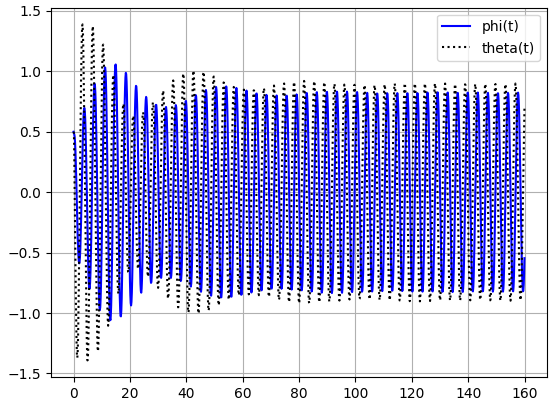} 
(b) \includegraphics[width=2.9in]{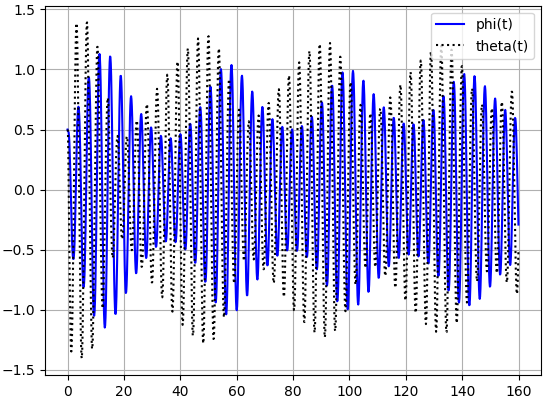}
\caption{\label{UnBalGraphs} Graphs of $\phi(t)$ (solid) and $\theta(t)$ (dotted) for unbalanced swing with the driven damped pendulum control for (a) $m_2=40$, $m_3=30$, $I_2=59.2$, $N=16$  and (b) $m_2=50$, $m_3=20$, $I_2=62.8$, $N=32$. Other parameters are $M=71$, $I_1=443.75$, $l_1=2.5$, $l_2=1$, $l_3=0.8$, $\alpha=100$, $\mu=-2$, $\o=1.98$, $\O=1.75$, $\kappa=100$, $r=1.55$.} 
\end{figure} 

From our experience with the balanced swing we should impose a limit cycle on $\psi$ (or its linear combination with $\phi$) and then consider the behavior of $\phi$ as driven by the resulting output. Supposing for simplicity (this corresponds to taking $\mu=0$ in Section ``From balanced swing to driven damped pendulum")
$$
\ddpsi+\kappa(\pfi-r^2)\dpsi+\O^2\psi=0\ \text{with}\ \pfi=\frac12\dpsi^2+\frac12\O^2\psi^2,
$$ 
as before, we obtain the control law
$$
v=-\O^2\psi-\kappa(\pfi-r^2)\dpsi+N\left(\frac{g}{I_1}\sin\left(\frac{I_1}{I_2}\psi\right)-\frac{l_1}{I_1}\sin\left(\frac{I_1}{I_2}\psi+\phi\right)\dphi^2\right).
$$
The closed loop equation for $\phi$ in the steady state is a perturbation of \eqref{chiBalSw} with $\mu=0$ and $\chi=-\frac{I_1}{I_2}\psi$:
\begin{equation}\label{UnBalSwPhiNeg}
\ddphi+c\dphi+\o^2\sin\phi\approx
-\O^2\psi+N\left(\frac{g}{I_1}\sin\left(\frac{I_1}{I_2}\psi\right)-\frac{l_1}{I_1}\sin\left(\frac{I_1}{I_2}\psi+\phi\right)\dphi^2\right).
\end{equation}
When $l_1\ll I_1=Ml_1^2$ and for small angles $\sin\psi\approx\psi$ this reduces to a harmonically driven damped pendulum, but with the driving amplitude corrected by an $N$ dependent term. 

Simulations implementing this control law are shown on Figure \ref{UnBalGraphs}.
For the smaller value of the imbalance $N=16$ we get regular swinging after a short transient, but when the imbalance is increased to $N=32$ this control produces beats despite the large baseline damping in the upper joint. At this value, the closed loop is an externally and parametrically driven pendulum, and the parametric effects can no longer be neglected when designing the control.

\section{Conclusions}\label{Conc}

We constructed two types of static state feedback torque controls for sustained oscillations of the balanced swing with a single dominant frequency. The first construction uses a combination of feedback linearization at zero damping in the upper joint, when the equations of motion coincide with those of the reaction wheel pendulum, with a novel static state feedback construction of a limit cycle for linear systems in the Brunovsky canonical form. The resulting feedback law performs well also for relatively small non-zero damping. The second construction is based on relating the balanced swing to harmonically driven damped pendulum, but the resulting control performs well (without beats and chaos) only for relatively large values. Both controls can be modified to shift the center of body rotations, so that the angles vary between two specified bounds. In particular, the angles can be kept positive to mimic human limitations and swings with a back bench. We also studied the relationship between these two controls in the limit of small swing amplitudes, and looked at the complications that arise when the swing is unbalanced. Balanced controls do seem to work for small imbalances, but physically (and mathematically) the unbalanced swing problem is much more complex.

To make analysis tractable, our approach was to find controls that allow to decompose the closed loop system into an autonomous subsystem with a stable limit cycle and a dependent subsystem driven by its output. Decompositions of this sort, other than the ones we found, would be useful, particularly those that behave well for the intermediate range of damping values. Energy functions other than the quadratic ones we used are also worth exploring. But even more desirable are alternative approaches to producing controls with analytically tractable closed loop systems, especially since human swinging may not be amenable to such a decomposition, see below.

\subsection{Feedback linearization}

It is of interest to investigate whether {\it partial} feedback linearization \cite{Spong,ZCM,Isid}, that was used for some control tasks for the Acrobot, can be adapted to generate limit cycles. Just like the Acrobot, the damped balanced swing has the relative degree three \cite[4.1]{Isid}, i.e. a three dimensional subsystem of it can be feedback linearized. The remaining degree of freedom can be made asymptotically stable, but it is not immediately clear what to do with it when sustained oscillations are the goal instead. The Theorem \ref{nDLimCyc} type construction does not seem to go through when the linearization is only partial.

Other variations on feedback linearization are of ineterst as well. A transformation that ``almost" feedback linearizes the Acrobot (with a small residual non-linearity) is defined in \cite{ZCM}. Perhaps it can be combined with a suitable generalization of our limit cycle construction to produce a torque control of the unbalanced swing without restrictions on the value of the imbalance. Approximate feedback linearization of the Acrobot in \cite{Mur} suggests another possible approach via hybrid controls.

\subsection{Control efficiency}

Our analysis had to be complemented by trial and error adjustment of control parameters to produce ``good behavior" for realistic swings. Part of the reason is that a more precise characterization of relevantly ``good" behavior is hard to spell out. Clearly, the resulting dynamics has to be bounded, but it is unclear how to define a  set of outputs that matches the intuition of ``regular oscillations" of the sort seen on a playground {\it and} has a precise analytic description with nice technical properties.

This also complicates studying control efficiency. We used efficiency considerations qualitatively to set up our design criteria and select parameter values, but a quantitative framework is much more preferable. Optimal control of the swing was considered for swinging in the standing position \cite{LaFo,Kul}, which is modeled by reducing the swing to a single mass moving up and down the swing axis, with its position as the control input. Extending the results to swinging in the sitting position, and with torques as control inputs, promises to be much more challenging even in the balanced case. It would be interesting to see what types of state constraints and cost functionals would produce optimal control with a limit cycle in the closed loop system, and whether we get ``regular oscillations" automatically, or some regularity constraint has to be imposed by hand. One possible choice is to include the Lyapunov function used in \eqref{LimCyc} into the cost functional. Heuristic considerations in \cite{WRR} suggest that for something like energy minimization with a lower bound on the amplitude the optimal control might be singular, with instant rotation at the turning points, but this already presupposes trajectories of a form for which the notion of ``amplitude" makes sense.

There is an additional complication in the unbalanced case. The two energy pumping mechanisms (center of mass motion and angular momentum transfer) interact in inscrutable ways, and work at cross purposes. The optimal time for the momentum transfer is at the turning points, while for the center of mass pumping it is a shift at the midswing \cite{WRR}. The most efficient control would have to be some non-linear mixture of the two, or, perhaps, even some quaint maneuvering that humans do not and/or can not perform. Although we took human swinging as the prototype, control parameters in our simulations can be selected so that the ``body" dumbbell oscillations are quite quirky compared to the ones observed on a playground. For example, the dumbbell may perform $180^\circ+$ flips near one of the turning points while the swing goes through regular oscillations unperturbed. This suggests that robotic devices designed to swing swings do not necessarily have to imitate humans, especially if energy efficiency, or some other technical considerations, are important.

\subsection{Human swinging}

In case one {\it is} interested in understanding the biophysics of human swinging, it is worth pointing out that even our shifted controls only partially ``fake" human behavior. To name a couple of discrepancies, there is no characteristic delay between reaching the turning point and the body's reaction, and the induced dumbbell rotations are too gradual. 

Ideally, one should instantaneously rotate the body (throw the head back and feet forward) from, say, $0^\circ$ to $90^\circ$  at the start of the forward swing, and rotate it back for the backward swing, instead of spreading out the motion harmonically over the entire half-swing, as in \eqref{Humth}. Human swingers are not ideal, and react with delay, but their repositioning still does not take the entire half-swing. The body is rotated swiftly near the turning points, and the swing oscillates naturally in between. This is similar to the workings of cogwheel escapements in old school mechanical clocks \cite{Bl1}. Since those are often modeled using piecewise smooth or van der Pohl limit cycles \cite{JK} those may provide a better fit for human swinging as well.

While human swinging maneuvers are simple enough to describe in terms of the body angle, reverse engineering torques that accomplish them is structurally more complex than the static state feedback philosophy we followed so far. We can implement the maneuvering described above by specifying the requisite torque $\tau$ for $\theta$. However, as the timing of the moves requires tracking $\dphi$, $\tau$ will depend on it, so $\tau=\tau(\dphi,\theta)$ and the equation for $\theta$ will not be autonomous. 

And this is not the end of it. As described, the procedure would pump energy into the swing endlessly resulting in unbounded growth of the swinging amplitude. Preventing that is simple enough for human swingers -- simply sit out a cycle when the swing up is too high -- then resume when damping reduces the amplitude sufficiently. However, implementing such a control directly requires memory, and the closed loop will no longer be described by a differential equation. A static state feedback workaround is to track swing's total energy instead, i.e. stop applying any torque whenever the energy exceeds a chosen threshold.

But even this static state feedback imitation of human control is not of the form that naturally flows from the standard methods of control theory. The reason is simple enough, the closed loop system does not have an autonomous subsystem, or some other apparent analytic backdoor, and is harder to analyze and predict as a result.
Our tentative attempts to experiment with human inspired controls in simulations were inconclusive because the relatively abrupt changes in the body angle produced torques that made the equation for swing angle numerically unstable.

\noindent{\large\bf Acknowledgments:} The authors are grateful to the anonymous reviewers for insightful critique of and many helpful suggestions to the original draft of the paper that greatly improved the final version.

\bibliographystyle{IEEEtran}
\bibliography{SwingRef}

\pagebreak
\thispagestyle{empty}

\end{document}